\pdfoutput=1
\documentclass[11pt,a4paper]{article}
\usepackage{amsfonts}
\usepackage{amsmath}
\usepackage{amssymb}
\usepackage{amsthm}
\usepackage{bbm}
\usepackage{microtype}      
\usepackage{color}
\usepackage{comment}
\usepackage{enumitem}
\usepackage[margin=2.41cm]{geometry}
\usepackage{graphicx}
\usepackage{hyperref}
\usepackage[utf8]{inputenc}
\usepackage{mathrsfs}
\usepackage{mathtools}
\usepackage{rotating}
\usepackage{tikz}
\usepackage{upgreek}
\usepackage{varwidth}
\usepackage[all,cmtip]{xy}
\usetikzlibrary{shapes.geometric}

\definecolor{darkred}{rgb}{0.8,0.1,0.1}
\hypersetup{
    colorlinks=false,         
    linkcolor=darkred,
    citecolor=blue,
}

\theoremstyle{plain}
\newtheorem{theo}{Theorem}[section]

\newtheorem{propo}[theo]{Proposition}

\theoremstyle{definition}
\newtheorem{defi}[theo]{Definition}

\newtheorem{exx}[theo]{Example}
\newenvironment{ex}
  {\pushQED{\qed}\exx}
  {\popQED\endexx}

\newtheorem{remm}[theo]{Remark}
\newenvironment{rem}
  {\pushQED{\qed}\remm}
  {\popQED\endremm}

\newtheorem{convv}[theo]{Convention}
\newenvironment{conv}
  {\pushQED{\qed}\convv}
  {\popQED\endconvv}

\numberwithin{equation}{section}
\def\nn{\nonumber}

\def\Mor{\mathrm{Mor}}

\def\cc{\mathrm{c}}
\def\colim{\mathrm{colim}}

\def\dd{\mathrm{d}}

\def\id{\mathrm{id}}
\def\oone{\mathbbm{1}}

\def\op{\mathrm{op}}

\def\supp{\mathrm{supp}}

\def\CC{\mathbf{C}}
\def\DD{\mathbf{D}}

\def\MM{\mathbf{M}}

\def\OO{\mathbf{O}}
\def\PPhi{\mathbf{\Phi}}

\def\AQFT{\mathbf{AQFT}}

\def\Desc{\mathbf{Desc}}

\def\Man{\mathbf{Man}}

\def\Mon{\mathbf{Mon}}
\def\Monr{\Mon_{\mathrm{rev}}}
\def\OCat{\mathbf{OrthCat}}

\def\Set{\mathbf{Set}}

\def\Vec{\mathbf{Vec}}
\def\VBun{\mathbf{VBun}}
\def\astAlg{\ast\mathbf{Alg}}

\def\astObj{\ast\mathbf{Obj}}
\def\astOp{\ast\mathbf{Op}}

\def\glue{\mathbf{glue}}
\def\data{\mathbf{data}}

\def\bbC{\mathbb{C}}

\def\bbR{\mathbb{R}}
\def\bbS{\mathbb{S}}

\def\bbZ{\mathbb{Z}}

\def\AAA{\mathfrak{A}}
\def\aaa{\mathfrak{a}}
\def\AAAA{\boldsymbol{\mathfrak{A}}}
\def\BBB{\mathfrak{B}}
\def\bbb{\mathfrak{b}}
\def\BBBB{\boldsymbol{\mathfrak{B}}}

\def\MMM{\mathcal{M}}
\def\O{\mathcal{O}}
\def\P{\mathcal{P}}

\def\sfE{\mathsf{E}}

\def\1{I}

\newcommand\und[1]{\underline{#1}}
\newcommand\ovr[1]{\overline{#1}}

\newcommand*\cocolon{       
        \nobreak
        \mskip6mu plus1mu
        \mathpunct{}%
        \nonscript
        \mkern-\thinmuskip
        {:}%
        \mskip2mu
        \relax
}

\def\sk{\vspace{1mm}}

\makeatletter
\let\@fnsymbol\@alph
\makeatother

\title{%
Gluing algebraic quantum field theories on manifolds
}

\author{%
Angelos Anastopoulos$^{1,a}$\ and\
Marco Benini$^{1,2,b}$\vspace{4mm}\\
{\small ${}^1$ Dipartimento di Matematica, Dipartimento di Eccellenza 2023-27, Universit\`a di Genova,}\\
{\small Via Dodecaneso 35, 16146 Genova, Italy.}\vspace{2mm}\\
{\small ${}^2$ INFN, Sezione di Genova,}\\
{\small Via Dodecaneso 33, 16146 Genova, Italy.}\vspace{4mm}\\
{\small \begin{tabular}{ll}
Email: & ${}^a$~\texttt{agganast@hotmail.com}\\
& ${}^b$~\texttt{marco.benini@unige.it}\vspace{2mm}
\end{tabular}
}
}

\date{December 2024}

\begin{document}

\maketitle


\begin{abstract}
\noindent It has been observed that, given an algebraic quantum field theory (AQFT) on a manifold $M$ and an open cover $\{M_\alpha\}$ of $M$, it is typically not possible to recover the global algebra of observables on $M$ by simply gluing the underlying local algebras subordinate to $\{M_\alpha\}$. Instead of gluing local algebras, we introduce a gluing construction for AQFTs subordinate to $\{M_\alpha\}$ and we show that for simple examples of AQFTs, constructed out of geometric data, gluing the local AQFTs subordinate to $\{M_\alpha\}$ recovers the global AQFT on $M$.
\end{abstract}




\renewcommand{\baselinestretch}{0.8}\normalsize
\tableofcontents
\renewcommand{\baselinestretch}{1.0}\normalsize

\newpage


\section{\label{sec:intro}Introduction}
Algebraic quantum field theory (AQFT) offers a mathematical axiomatization 
of various flavours of quantum field theory, most notably 
on Lorentzian manifolds (AQFT \`a la Haag--Kastler 
\cite{HaagKastler_1964_AlgebraicApproach} and locally covariant QFT
\cite{BrunettiFredenhagenVerch_2003_GenerallyCovariant}, 
including their conformal variants \cite{Rehren_2000_ChiralObservables,Pinamonti_2009_ConformalGenerally,BeniniGiorgettiSchenkel_2022_SkeletalModel}), 
but also on Riemannian manifolds \cite{Schlingemann_1999_EuclideanField}, 
all of which share the following three crucial features: 
(1)~physical observables are encompassed 
by non-commutative $\ast$-algebras, 
(2)~the dependence of observables on spacetime regions 
is encoded by a functor from a suitable source category 
(of spacetime regions) to the category of $\ast$-algebras 
(of observables), 
(3)~the causality axiom forces the algebraic multiplication 
to be ``commutative'' when the inputs come from 
distinguished pairs of morphisms in the source category
(such as causally disjoint spacetime subregions, 
i.e.\ that can only be joined by curves 
that exceed the speed of light). 
\sk 

Frequently, items~(1) and~(3) receive more attention 
compared to item~(2). Indeed, item~(1) provides 
an evident connection to the realm of quantum physics. 
Yet, algebras of observables are often mathematically 
indistinguishable, namely isomorphic, 
for physical models that are actually quite different. 
Item~(2) encodes more refined information that 
makes it possible to distinguish them also mathematically. 
(For instance, on each spacetime region, the algebras of observables 
associated with free scalar fields of different masses 
are isomorphic, 
however these isomorphisms are not natural, 
i.e.\ the corresponding functors assigning algebras 
to spacetime regions are not naturally isomorphic.) 
Finally, item~(3) is perhaps the main selling point
of AQFT from the physical point of view, 
because it codifies the fact that operations performed 
in causally disjoint spacetime region cannot affect each other 
(loosely speaking, no physical effect propagates faster than light). 
Yet, it is remarkable that, the statement of item~(3) 
becomes completely meaningless neglecting either item~(1) or item~(2). 
\sk 

Another situation where neglecting part of the AQFT structure 
leads to unpleasant consequences is related to descent, namely 
the attempt to reconstruct global information on a manifold $M$ 
by gluing local information subordinate to an open cover of $M$. 
(From a physical point of view, 
one may interpret descent as a way to reconstruct 
a large physical system by merging smaller parts of the system, 
which are potentially simpler to describe.) 
Since AQFTs are in particular $\ast$-algebra-valued functors $\AAA$ 
on a category of geometric objects (typically manifolds 
endowed with additional structure), it seems natural to attempt 
to formalize descent for AQFTs as a cosheaf condition 
on the underlying functor. 
Intuitively, one would expect descent to hold 
at least for those AQFTs that are constructed out of local data 
(such as vector bundles, connections, etc). 
On the contrary, even for the simple AQFT $\AAA$ 
associated with a real scalar field on the circle $\bbS^1$ 
the above mentioned cosheaf condition is violated. Indeed, 
it turns out \cite[Appendix]{BeniniSchenkel_2019_HigherStructures} 
(see also \ref{th:alg}) 
that the global algebra $\AAA(\bbS^1)$ 
is not recovered by gluing the local algebras 
$\AAA(M_\alpha)$ on an open cover $\{M_\alpha\}$ of $\bbS^1$
(unless $\bbS^1$ itself belongs to the cover). 
\sk

In front of this evidence, one is forced to reconsider 
the former attempt to formalize descent for AQFTs. 
Evidently, the above cosheaf condition ignores part of the AQFT structure 
(item (2) and consequently also item (3) above). Instead, 
when the gluing construction retains also this information, 
the situation improves considerably, as Theorem \ref{th:AQFT} shows: 
gluing the local AQFTs $\AAA \vert_{M_\alpha}$ on an open cover 
$\{M_\alpha\}$ of $\bbS^1$ recovers the global AQFT 
$\AAA$ of a real scalar field on $\bbS^1$. 
\sk 

Let us now explain how this result is achieved 
by developing a suitable gluing construction for AQFTs. 

In order to do so, let us start from the simpler problem 
of gluing the underlying local algebras 
on an open cover $\{M_\alpha\}$ of a manifold $M$. 
The latter problem boils down to computing the colimit of the diagram 
\begin{equation}\label{eq:alg-glue}
\xymatrix@C=1.5em{
\coprod_{\alpha\beta} \AAA(M_{\alpha\beta}) \ar@<2pt>[r] \ar@<-2pt>[r] & \coprod_\alpha \AAA(M_\alpha) 
}
\end{equation}
in the category of $\ast$-algebras. 

In order to glue the local AQFTs $\AAA \vert_{M_\alpha}$, 
one has to compute a colimit of the same shape, 
however in the category of AQFTs over $M$: 
\begin{equation}\label{eq:AQFT-glue}
\xymatrix{
\coprod_{\alpha\beta} \AAA \vert_{M_{\alpha\beta}} \vert^M \ar@<2pt>[r] \ar@<-2pt>[r] & \coprod_\alpha \AAA \vert_{M_\alpha} \vert^M 
} .  
\end{equation}

In spite of the formal analogy, the computation 
of the colimit \eqref{eq:AQFT-glue} is considerably 
more involved than that of the colimit \eqref{eq:alg-glue} 
because AQFTs have a much richer structure 
compared to $\ast$-algebras. 
Informally speaking, one must face two complications: 
\begin{itemize}
\item the superficial one is that, instead of dealing, 
for each patch $M_\alpha$, 
with a single $\ast$-algebra $\AAA(M_\alpha)$, 
one has to deal with a functor $\AAA \vert_{M_\alpha}$ assigning 
$\ast$-algebras to each subregion of $M_\alpha$; 

\item the deeper (and more challenging) one is that 
the gluing construction must be compatible with the causality axiom 
(item (3) above). 
\end{itemize}

It is the latter complication that prevents one from performing 
the gluing construction one spacetime region at a time, 
i.e.\ merely gluing the functors $\AAA \vert_{M_\alpha}$ 
object-wise. Instead, one has to tackle the problem ``globally'', 
which can be done by hard-coding the causality axiom 
as a structure, rather than a property. This is achieved 
by regarding AQFTs as $\ast$-algebras over the AQFT colored $\ast$-operad 
\cite{BeniniSchenkelWoike_2019_InvolutiveCategories,BeniniSchenkelWoike_2021_OperadsAlgebraica,Yau_2020_InvolutiveMonoidal}. 
Then the theory of operads, see e.g.\ 
\cite{Fresse_2009_ModulesOperads,Yau_2016_ColoredOperads}, 
provides means for computing the relevant colimit. 
This leads to the gluing construction for AQFTs 
that we develop in Subsection \ref{subsec:gluing}. 

Using our gluing construction for AQFTs, we show that 
gluing the local AQFTs $\AAA \vert_{M_\alpha}$ recovers 
the global AQFT $\AAA$ of a real scalar field on the circle $\bbS^1$, 
in sharp contrast with the failure to recover 
the global algebra $\AAA(\bbS^1)$ on the circle 
by gluing the local algebras $\AAA(M_\alpha)$ on the cover. 
\sk 

Before we move on to a more detailed outline of the content of this paper, 
let us mention how our gluing construction relates to the so-called 
Fredenhagen's universal algebra, a local-to-global construction 
that is well-known in the context of algebraic quantum field theory. 
The original form of Fredenhagen's universal algebra 
\cite{Fredenhagen_1993_GlobalObservables} can be interpreted 
as the left Kan extension of the functor underlying an AQFT $\AAA$ 
defined on a full subcategory of ``small'' spacetime regions 
(whose choice depends on the nature of the problem under investigation) 
along the embedding into the category of all spacetime regions. 
The shortcoming of this approach is that such a construction 
forgets the causality axiom and hence it does not guarantee that 
the resulting left Kan extended functor fulfils the causality axiom. 
This issue has been solved in \cite{BeniniSchenkelWoike_2021_OperadsAlgebraica} 
by refining Fredenhagen's universal algebra to {\it operadic} 
left Kan extension, which keeps track of the causality axiom 
and hence produces a genuine AQFT defined on all spacetimes 
out of an AQFT defined on ``small'' spacetimes. 
The key difference between the refined Fredenhagen's universal algebra 
construction and the gluing construction presented in this paper 
lies already in the respective inputs: while the former takes as input 
a single AQFT defined on ``small'' regions, the latter takes as input 
a compatible family of AQFTs subordinate to a cover. 
On the one hand, Fredenhagen's universal algebra produces the universal AQFT 
that to the ``small'' regions assigns algebras of observables 
that agree with those assigned by the input AQFT. 
On the other hand, our gluing construction produces the universal AQFT that, 
upon restriction to the given cover, agrees with the input family of AQFTs. 
\sk

Let us now outline the content of the rest of the paper. 
Section \ref{sec:prel} summarizes the mathematical tools 
that are needed to formalize AQFTs as $\ast$-algebras 
over a suitable colored $\ast$-operad, 
to develop our gluing construction for AQFTs and, 
in particular, to compute the relevant colimits 
in the category of operad algebras. 
Building upon this preliminary material, 
Section \ref{sec:fibcat} develops the gluing construction 
for AQFTs. As a preliminary step, Subsection \ref{subsec:fib-cat} 
assembles AQFTs over any manifold $M$ into a fibred category, 
which provides the necessary context to develop 
our gluing construction in Subsection \ref{subsec:gluing}. 
The latter is used in Section \ref{sec:descent}. 
First, Subsection \ref{subsec:probe} presents 
a class of AQFTs on manifolds (of a fixed odd dimension) 
that are constructed out of a geometric input, 
consisting of a natural vector bundle endowed with 
a natural fiber metric and a natural metric connection. 
(The $1$-dimensional real scalar field falls within this class.) 
Second, Subsection \ref{subsec:alg-descent} shows that 
gluing local algebras on a generic open cover fails to recover 
the global algebra on the manifold for this class of AQFTs 
(Theorem \ref{th:alg}). (This is just a mild generalization 
of the similar observation from 
\cite[Appendix]{BeniniSchenkel_2019_HigherStructures}.) 
In contrast with the previous failure, 
Subsection \ref{subsec:aqft-descent} proves that 
gluing local AQFTs recovers the global AQFT 
for the class of AQFTs mentioned above (Theorem \ref{th:AQFT}).


\section{\label{sec:prel}Preliminaries}

\subsection{\label{subsec:ISMCat}Involutive symmetric monoidal categories}
In this section we recall the concept of an involutive 
symmetric monoidal category and the associated symmetric monoidal 
category of $\ast$-objects. The latter will provide the base 
category for colored $\ast$-operads and $\ast$-algebras over them, 
which will be recalled in Section~\ref{subsec:operads} below. 
More details on involutive category theory, including proofs 
of the facts recalled in the present section, can be found in 
\cite{Jacobs_2012_InvolutiveCategories} 
and also in \cite{BeniniSchenkelWoike_2019_InvolutiveCategories}. 

\begin{defi}\label{defi:ISMCat}
An {\it involutive symmetric monoidal category} 
$\MM = (\MM,\otimes,\oone,J,j)$ consists of a symmetric monoidal 
category $(\MM,\otimes,\oone)$, a symmetric monoidal endofunctor 
$J = (J,J_2,J_0)\nobreak\mskip2mu\mathpunct{}\nonscript\mkern-\thinmuskip{:} \linebreak[0] (\MM,\otimes,\oone) \to (\MM,\otimes,\oone)$ 
and a symmetric monoidal natural isomorphism $j\colon \id \to J^2$ 
such that $j J = J j$.
\end{defi}

\begin{propo}
Let $\MM$ be an involutive symmetric monoidal category. 
Then the underlying symmetric monoidal endofunctor $J$ 
is strong (i.e.\ $J_2$ and $J_0$ are isomorphisms) 
and self-adjoint (more precisely, the natural isomorphism $j$ 
witnesses an adjunction $J \dashv J$). 
In particular, $J$ preserves both limits and colimits. 
\end{propo}

Iterating $J_2$ allows one to define comparison isomorphisms $J_n$
involving $n$-fold tensor products, instead of just binary ones. 
To an involutive symmetric monoidal category one can canonically 
assign a symmetric monoidal category of $\ast$-objects, 
whose definition is recalled below. 

\begin{defi}\label{defi:astObj}
Let $\MM$ be an involutive symmetric monoidal category. 
A {$\ast$-object} $V$ in $\MM$ consists of an object $V \in \MM$ 
and a {\it $\ast$-involution} $\ast_V\colon V \to J(V)$ in $\MM$ 
such that $J(\ast) \circ \ast = j_V$. 
A {$\ast$-morphism} $f\colon V \to W$ in $\MM$ consists of a morphism 
$f\colon V \to W$ in $\MM$ such that $J(f) \circ \ast_V = \ast_W \circ f$. 
The category $\astObj(\MM)$ has $\ast$-objects in $\MM$ as objects, 
$\ast$-morphisms in $\MM$ as morphisms and it is endowed 
with the symmetric monoidal structure induced 
by the symmetric monoidal structures of $\MM$ and $J$. 
\end{defi}

It follows from the definition of a $\ast$-object 
that the underlying $\ast$-involution is an isomorphism. 
Furthermore, let us recall with the next statement 
that the symmetric monoidal category $\astObj(\MM)$ 
inherits nice properties from $\MM$. 

\begin{propo}\label{propo:astObj}
Let $\MM$ be an involutive symmetric monoidal category. 
If $\MM$ is closed, then $\astObj(\MM)$ is closed too. 
Furthermore, $\astObj(\MM)$ has all limits and colimits 
that exist in the category $\MM$ and those are created 
by the forgetful functor $\astObj(\MM) \to \MM$. 
\end{propo}

\begin{ex}\label{ex:Vec}
For our purposes the prime example of an involutive 
symmetric monoidal category is the category $\Vec_{\bbC}$ 
of $\bbC$-vector spaces, endowed 
with the standard symmetric monoidal structure 
and the involutive structure consisting of the 
symmetric monoidal endofunctor 
$\ovr{(-)}\colon \Vec_{\bbC} \to \Vec_{\bbC}$ 
that to a vector space $V \in \Vec_{\bbC}$ assigns its complex 
conjugate $\ovr{V} \in \Vec_{\bbC}$ and the symmetric monoidal 
natural isomorphism $\id\colon \id \to \ovr{\ovr{(-)}}$. 
The symmetric monoidal category $\astObj(\Vec_{\bbC})$ 
of $\ast$-objects in $\Vec_{\bbC}$ consists of $\bbC$-vector spaces 
$V$ endowed with a $\ast$-involution $\ast_V\colon V \to \ovr{V}$, 
namely an anti-linear map that squares to the identity, 
and $\bbC$-linear maps $f\colon V \to W$ compatible with the 
$\ast$-involutions. For instance, the complex conjugation map 
$\ovr{(-)}\colon \bbC \to \ovr{\bbC}$ in $\Vec_{\bbC}$ 
is a $\ast$-object in $\Vec_{\bbC}$. 
Note that $\Vec_{\bbC}$ is closed and cocomplete, 
hence $\astObj(\Vec_{\bbC})$ is closed and cocomplete too. 
\end{ex}

\subsection{\label{subsec:operads}Colored \texorpdfstring{$\ast$}{star}-operads and \texorpdfstring{$\ast$}{star}-algebras over colored \texorpdfstring{$\ast$}{star}-operads}
In this section we recall the basic tools from the theory of 
colored $\ast$-operads and $\ast$-algebras over them 
that will be needed in the rest of the paper. 
The material presented here is covered by 
\cite{BeniniSchenkelWoike_2019_InvolutiveCategories,BeniniSchenkelWoike_2021_OperadsAlgebraica} 
and by the recent textbook \cite{Yau_2020_InvolutiveMonoidal}, 
the only exception being an explicit model for computing colimits 
in the category of $\ast$-algebras over a colored $\ast$-operad. 
Standard textbook references for the general theory 
of colored operads (also known as multicategories) and algebras over them 
are \cite{Fresse_2009_ModulesOperads,Yau_2016_ColoredOperads}. 

Let us fix an involutive symmetric monoidal category $\MM = (\MM,\otimes,\oone,J,j)$. 
Unless otherwise stated, from now on we assume that $\MM$ 
is both closed and cocomplete. 

\begin{defi}\label{defi:operad}
A colored $\ast$-operad $\O$ in $\MM$ is a colored operad in 
$\astObj(\MM)$. Explicitly, $\O$ consists of the data listed below: 
\begin{enumerate}[label=(\alph*)]
\item A set of {\it colors} $c \in \O$. 
\item For all $t \in \O$ and all tuples 
$\und{c} = (c_1, \ldots, c_n)$ in $\O$ of length $n \geq 0$, 
an object of {\it operations} $\O(\substack{t\\ \und{c}}) \in \MM$.
\item For all $t \in \O$, all tuples $\und{b}$ in $\O$ 
of length $m \geq 1$ and all tuples $\und{a}_i$ in $\O$ of length 
$k_i \geq 0$, $i=1,\ldots,m$, a {\it composition} morphism 
\begin{equation}
\gamma\colon \O(\substack{t\\ \und{b}}) \otimes \bigotimes_{i=1}^m \O(\substack{b_i\\ \und{a}_i}) \longrightarrow \O(\substack{t\\ \und{\und{a}}})
\end{equation}
in $\MM$, where $\und{\und{a}} \coloneqq (\und{a}_1,\ldots,\und{a}_m)$ 
denotes the tuple formed by juxtaposition. 
\item For all $t \in \O$, an {\it identity} morphism 
\begin{equation}
\id_t\colon \oone \longrightarrow \O(\substack{t\\ t})
\end{equation}
in $\MM$. 
\item For all $t \in \O$, all tuples $\und{c}$ in $\O$ 
of length $n \geq 0$ and all permutations $\sigma \in \Sigma_n$, 
a morphism 
\begin{equation}
\O(\sigma)\colon \O(\substack{t\\ \und{c}}) \longrightarrow \O(\substack{t\\ \und{c}\sigma})
\end{equation}
in $\MM$, where 
$\und{c}\sigma \coloneqq (c_{\sigma(1)},\ldots,c_{\sigma(n)})$ 
is given by right permutation. 
\item For all $t \in \O$ and all tuples $\und{c}$ in $\O$ 
of length $n \geq 0$, a {\it $\ast$-involution} 
\begin{equation}
\ast\colon \O(\substack{t\\ \und{c}}) \longrightarrow J(\O(\substack{t\\ \und{c}}))
\end{equation}
in $\MM$.
\end{enumerate}
The data above are subject to the following axioms: 
\begin{enumerate}[label=(\roman*)]
\item The data~(c--e) are subject to the usual colored operad axioms. 
\item The data~(f) are subject to the $\ast$-object axiom, 
see Definition~\ref{defi:astObj}. 
\item The data~(c--e) and (f) are subject to compatibility axioms. 
Explicitly, the $\ast$-involutions from (f) must be compatible 
with:
\begin{itemize}
    \item the compositions from (c), namely the diagram 
    \begin{subequations}
    \begin{equation}
    \xymatrix{
    \O(\substack{t\\ \und{b}}) \otimes \displaystyle \bigotimes_{i=1}^m \O(\substack{b_i\\ \und{a}_i}) \ar[rr]^-{\gamma} \ar[d]_-{\ast \otimes \bigotimes_i \ast} && \O(\substack{t\\ \und{\und{a}}}) \ar[d]^-{\ast} \\ 
    J(\O(\substack{t\\ \und{b}})) \otimes \displaystyle \bigotimes_{i=1}^m J(\O(\substack{b_i\\ \und{a}_i})) \ar[r]_-{J_{m+1}} & J \left( \O(\substack{t\\ \und{b}}) \otimes \displaystyle \bigotimes_{i=1}^m \O(\substack{b_i\\ \und{a}_i}) \right) \ar[r]_-{J(\gamma)} & J(\O(\substack{t\\ \und{\und{a}}}))
    }
    \end{equation}
    in $\MM$ commutes for all $t \in \O$, all tuples $\und{b}$ in $\O$ 
    of length $m \geq 1$ and all tuples $\und{a}_i$ in $\O$ 
    of length $k_i \geq 0$, $i=1,\ldots,m$.
    \item the identities from (d), namely the diagram 
    \begin{equation}
    \xymatrix{
    \oone \ar[rr]^-{\id_t} \ar[d]_-{J_0} && \O(\substack{t\\ t}) \ar[d]^-{\ast} \\ 
    J(\oone) \ar[rr]_-{J(\id_t)} && J(\O(\substack{t\\ t}))
    }
    \end{equation}
    in $\MM$ commutes for all $t \in \O$.
    \item the permutation actions from (e), namely the diagram 
    \begin{equation}
    \xymatrix{
    \O(\substack{t\\ \und{c}}) \ar[rr]^-{\O(\sigma)} \ar[d]_-{\ast} && \O(\substack{t\\ \und{c}\sigma}) \ar[d]^-{\ast} \\ 
    J(\O(\substack{t\\ \und{c}})) \ar[rr]_-{J(\O(\sigma))} && J(\O(\substack{t\\ \und{c}\sigma}))
    }
    \end{equation}
    \end{subequations}
    in $\MM$ commutes for all $t \in \O$, all tuples $\und{c}$ in $\O$ 
    of length $n \geq 0$ and all permutations $\sigma \in \Sigma_n$. 
\end{itemize}
\end{enumerate}
\end{defi}

\begin{defi}\label{defi:multifunctor}
A $\ast$-multifunctor $F\colon \O \to \P$ between the colored 
$\ast$-operads $\O, \P$ in $\MM$ is a multifunctor $F\colon \O \to \P$ 
between the colored operads $\O, \P$ in $\astObj(\MM)$. 
Explicitly, $F$ consists of the data listed below: 
\begin{enumerate}[label=(\alph*)]
\item A map $F\colon \O \to \P$ between the underlying sets of colors. 
\item For all $t \in \O$ and all tuples $\und{c}$ in $\O$ 
of length $n \geq 0$, a morphism 
$F\colon \O(\substack{t\\ \und{c}}) \to \P(\substack{Ft \\ F\und{c}})$ 
in $\MM$, where $F\und{c} \coloneqq (Fc_1,\ldots,Fc_n)$ is the tuple 
formed by acting with $F$ on each component. 
\end{enumerate}
The data above are subject to the evident compatibilities 
with compositions, identities, permutation actions 
and $\ast$-involutions of $\O$ and $\P$. 
\end{defi} 

\begin{defi}
The category $\astOp$ has colored $\ast$-operads $\O$, 
see Definition~\ref{defi:operad}, as objects 
and $\ast$-multifunctors $F\colon \O \to \P$, see Definition~\ref{defi:multifunctor}, as morphisms. 
\end{defi}

\begin{defi}\label{defi:astAlg}
A $\ast$-algebra $A$ over a colored $\ast$-operad 
$\O \in \astOp$ in $\MM$ is an algebra over the colored operad $\O$ 
in $\astObj(\MM)$. Explicitly, $A$ consists of the data listed below:
\begin{enumerate}[label=(\alph*)]
\item For all $t \in \O$, an object $A(t) \in \MM$, 
\item For all $t \in \O$ and all tuples $\und{c}$ in $\O$ 
of arbitrary length $n \geq 0$, an {\it $\O$-action} 
\begin{equation}
\alpha\colon \O(\substack{t\\ \und{c}}) \otimes \displaystyle\bigotimes_{i=1}^n A(c_i) \longrightarrow A(t)
\end{equation}
in $\MM$. 
\item For all $t \in \O$, a {\it $\ast$-involution} 
\begin{equation}
\ast_A\colon A(t) \longrightarrow J(A(t))
\end{equation}
in $\MM$. 
\end{enumerate}
These data are subject to the following axioms: 
\begin{enumerate}[label=(\roman*)]
\item The data~(a-b) are subject to the usual axioms 
for algebras over the colored operad $\O$. 
\item The data~(c) are subject to the $\ast$-object axiom, 
see Definition~\ref{defi:astObj}.
\item The data~(b) and (c) are subject to compatibility axioms. 
Explicitly, the $\O$-actions from (b) must be compatible with 
the $\ast$-involutions from (c) and those underlying the 
$\ast$-operad $\O$, namely the diagram 
\begin{equation}
\xymatrix{
\O(\substack{t\\ \und{c}}) \otimes \displaystyle\bigotimes_{i=1}^n A(c_i) \ar[rr]^-{\alpha} \ar[d]_-{\ast \otimes \bigotimes_i \ast_A} && A(t) \ar[d]^-{\ast_A} \\ 
J(\O(\substack{t\\ \und{c}})) \otimes \displaystyle\bigotimes_{i=1}^n J(A(c_i)) \ar[r]_-{J_{n+1}} & J \left( \O(\substack{t\\ \und{c}}) \otimes \displaystyle\bigotimes_{i=1}^n A(c_i) \right) \ar[r]_-{J(\alpha)} & J(A(t))
}
\end{equation}
in $\MM$ commutes for all $t \in \O$ and all tuples $\und{c}$ 
in $\O$ of length $n \geq 0$. 
\end{enumerate}
\end{defi}

\begin{defi}\label{defi:astAlgmor}
A {\it $\ast$-algebra morphism} $\Phi\colon A \to B$ over a colored 
$\ast$-operad $\O \in \astOp$ in $\MM$ is an algebra morphism 
$\Phi\colon A \to B$ over the colored operad $\O$ in $\astObj(\MM)$. 
Explicitly, $\Phi$ consists of a morphism $\Phi_c\colon A(c) \to B(c)$ 
in $\MM$ for all $c \in \O$, subject to the evident compatibilities 
with $\O$-actions and $\ast$-involutions of $A$ and $B$. 
\end{defi}

\begin{defi}
The category $\astAlg(\O)$, for a colored $\ast$-operad 
$\O \in \astOp$, has $\ast$-algebras $A$ 
over $\O$, see Definition~\ref{defi:astAlg}, as objects 
and $\ast$-algebra morphisms $\Phi\colon A \to B$ over $\O$, 
see Definition~\ref{defi:astAlgmor}, as morphisms. 
\end{defi}

Let us recall that a $\ast$-multifunctor $F\colon \O \to \P$ 
gives rise to an adjunction 
\begin{equation}\label{eq:F!F*}
\xymatrix{
F_!\colon \astAlg(\O) \ar@<2pt>[r] & \astAlg(\P) \cocolon F^\ast. \ar@<2pt>[l]
}
\end{equation}
between the categories of $\ast$-algebras over $\O$ and $\P$. 
The right adjoint functor $F^\ast$ is simply given by restriction 
along $F$. Its left adjoint $F_!$ may be interpreted as 
a ``$\ast$-operadic'' left Kan extension. 
The adjunction \eqref{eq:F!F*} enhances the similar one for algebras
over colored operads by taking $\ast$-involutions into account. 

One recovers the free-forgetful adjunction 
\begin{equation}\label{eq:free-forget}
\xymatrix{
L\colon \displaystyle\prod_{t \in \O} \astObj(\MM) \ar@<2pt>[r] & \astAlg(\O) \cocolon R. \ar@<2pt>[l]
}
\end{equation} 
for $\ast$-algebras over $\O \in \astOp$ as a special instance 
of the adjunction \eqref{eq:F!F*} as follows. 
Consider the 
colored $\ast$-operad $\O_0 \in \astOp$ on the set of colors 
of $\O$ and with operations given by the initial object 
$\O_0(\substack{t\\ \und{c}}) \coloneqq \emptyset \in \MM$ for all 
$t \in \O_0$ and all tuples $\und{c} \neq (t)$ in $\O_0$, and $\O_0(\substack{t\\ t}) \coloneqq \oone \in \MM$ for all $t \in \O_0$. 
(The identities are given by the identity $\id_\oone$ morphism on $\oone$, while compositions, permutation actions and 
$\ast$-involutions of $\O_0$ are necessarily trivial.) 
Consider also the $\ast$-multifunctor $F\colon \O_0 \to \O$ 
in $\astOp$ given by the identity on the underlying sets of colors. 
(The action $F \colon \O_0(\substack{t\\ \und{c}}) \to \O(\substack{Ft\\ F\und{c}})$ is given 
by the identity morphism $\id_{Ft} = \id_\oone \colon \oone \to \O(\substack{Ft\\ Ft})$ when $\und{c} = (t)$, 
whereas it is necessarily 
trivial otherwise.) Then, by indentifying canonically the categories 
$\astAlg(\O_0) = \prod_{t \in \O}\astObj(\MM)$ of $\ast$-algebras over 
$\O_0 \in \astOp$ and of colored collections of $\ast$-objects in $\MM$, one obtains 
the adjunction \eqref{eq:free-forget} by specializing 
\eqref{eq:F!F*} to the $\ast$-multifunctor $F\colon \O_0 \to \O$ 
in $\astOp$ described above. One immediately realizes that 
the right adjoint functor $R$ acts on a $\ast$-algebra 
$A \in \astAlg(\O)$ simply by forgetting its $\O$-action. 

Since it will be later useful in order to present 
an explicit model computing a colimit in $\astAlg(\O)$, 
let us also recall the explicit form of the left adjoint functor 
$L$, as well as that of the adjunction counit $\varepsilon$. 
The left adjoint $L$ assigns to a colored collection of 
$\ast$-objects $V_\bullet = (V_t) \in \prod_{t \in \O} \astObj(\MM)$ 
the $\ast$-algebra $L(V_\bullet) \in \astAlg(\O)$ 
consisting of the data listed below: 
\begin{enumerate}[label=(\alph*)]
\item For all $t \in \O$, the object 
\begin{subequations}
\begin{equation}
L(V_\bullet)(t) \coloneqq \coprod_{n \geq 0} L^n(V_\bullet)(t) \in \MM,
\end{equation}
where 
\begin{equation}
L^n(V_\bullet)(t) \coloneqq \colim \left( \xymatrix@C=3.5em{\displaystyle\coprod_{(c_1,\ldots,c_n)} \displaystyle\coprod_{\sigma \in \Sigma_n} \O(\substack{t\\ \und{c}}) \otimes \displaystyle\bigotimes_{i=1}^n V_{c_{\sigma(i)}} \ar@<2pt>[r]^-{\O(\sigma) \otimes \id} \ar@<-2pt>[r]_-{\id \otimes b_{\sigma}} & \displaystyle\coprod_{(c_1,\ldots,c_n)} \O(\substack{t\\ \und{c}}) \otimes \displaystyle\bigotimes_{i=1}^n V_{c_i}} \right) \in \MM 
\end{equation}
\end{subequations}
is defined as the coequalizer of the parallel pair given by the 
permutation action of $\O$ and the symmetric braiding of $\MM$. 
\item For all $t \in \O$ and all tuples $\und{b}$ in $\O$ 
of length $m \geq 0$, the $\O$-action 
\begin{subequations}
\begin{equation}
\lambda \colon 
\xymatrix{\O(\substack{t\\ \und{b}}) \otimes \displaystyle\bigotimes_{i=1}^m L(V_\bullet)(b_i) \ar[r] & L(V_\bullet)(t)}
\end{equation}
in $\MM$, which is defined by the diagram
\begin{equation}\label{eq:free-O-action}
\xymatrix{
\O(\substack{t\\ \und{b}}) \otimes \displaystyle\bigotimes_{i=1}^m L(V_\bullet)(b_i) \ar@{-->}[r]^-{\lambda} & L(V_\bullet)(t) \\ 
\O(\substack{t\\ \und{b}}) \otimes \displaystyle\bigotimes_{i=1}^m \left( \O(\substack{b_i\\ \und{a}_i}) \otimes \displaystyle\bigotimes_{j=1}^{k_i} V_{a_{ij}} \right) \ar[u]^-{\id \otimes \bigotimes_i \iota_{k_i,\und{a}_i}} \\ 
\O(\substack{t\\ \und{b}}) \otimes \displaystyle\bigotimes_{i=1}^m \O(\substack{b_i\\ \und{a}_i}) \otimes \displaystyle\bigotimes_{i=1}^m \displaystyle\bigotimes_{j=1}^{k_i} V_{a_{ij}} \ar[r]_-{\gamma \otimes \id} \ar[u]^-{\cong} & \O(\substack{t\\ \und{\und{a}}}) \otimes \displaystyle\bigotimes_{i=1}^m \displaystyle\bigotimes_{j=1}^{k_i} V_{a_{ij}} \ar[uu]_-{\iota_{\sum_i k_i,\und{\und{a}}}}
}
\end{equation}
\end{subequations}
in $\MM$, for all $k_1, \ldots, k_m \geq 0$ and all tuples 
$\und{a_i}$ in $\O$ of length $k_i$. In the above diagram the 
vertical isomorphism is given by the symmetric braiding of $\MM$, 
whereas by $\gamma$ we denote the composition of $\O$ and by 
$\iota_{n,\und{c}}$, for some $n \geq 0$ and some tuple $\und{c}$ in $\O$ of length $n$, we denote the morphism
\begin{equation}
\xymatrix{
\O(\substack{t\\ \und{c}}) \otimes \displaystyle\bigotimes_{i=1}^n V_{c_i} \ar[r]^-{l_{\und{c}}} & L^n(V_\bullet)(t) \ar[r]^-{\iota_n} & L(V_\bullet)(t)},
\end{equation}
where $l_{\und{c}}$ is determined by the inclusion in the $\und{c}$-component of the coproduct $\coprod_{\und{d}} \O(\substack{t\\ \und{d}}) \otimes \bigotimes_{i=1}^n V_{d_i}$ and the respective leg of the limit cocone to $ L^n(V_\bullet)(t)$, while $\iota_n$ denotes the inclusion in the $n$-component of the coproduct $\coprod_{\ell \geq 0} L^\ell(V_\bullet)(t) = L(V_\bullet)(t)$.
\item For all $t \in \O$, the $\ast$-involution 
\begin{subequations}
\begin{equation}
\ast_L \colon 
\xymatrix{L(V_\bullet)(t) \ar[r] & J(L(V_\bullet)(t))}
\end{equation}
in $\MM$, which is defined by the diagram 
\begin{equation}\label{eq:free-O-involution}
\xymatrix@C=4em{
L(V_\bullet)(t) \ar@{-->}[r]^-{\ast_L} & J(L(V_\bullet)(t)) \\ 
& J \left( \O(\substack{t\\ \und{c}}) \otimes \displaystyle\bigotimes_{i=1}^n V_{c_i} \right) \ar[u]_-{J(\iota_{n,\und{c}})} \\ 
\O(\substack{t\\ \und{c}}) \otimes \displaystyle\bigotimes_{i=1}^n V_{c_i} \ar[r]_-{\ast \otimes \bigotimes_i \ast_{V_{c_i}}} \ar[uu]^-{\iota_{n,\und{c}}} & J(\O(\substack{t\\ \und{c}})) \otimes \displaystyle\bigotimes_{i=1}^n J(V_{c_i}) \ar[u]_-{J_{n+1}}
}
\end{equation}
\end{subequations}
in $\MM$ involving the $\ast$-involutions 
$\ast$ of $\O$ and $\ast_{V_c}$ of $V_c$ for all $c \in \O$. 
\end{enumerate} 

Concerning the counit 
\begin{equation}\label{eq:free-forget-counit}
\varepsilon\colon L R \longrightarrow \id 
\end{equation}
of the free-forgetful adjunction $L \dashv R$ 
from \eqref{eq:free-forget}, it is defined component-wise, 
for all $\ast$-algebras $A \in \astAlg(\O)$, by setting $\varepsilon_A\colon L R(A) \longrightarrow A$ in $\astAlg(\O)$
to be given by the $\O$-action $\alpha$ underlying $A$. 
\sk

Later on, we shall compute certain colimits in $\astAlg(\O)$, 
for a given colored $\ast$-operad $\O \in \astOp$ in $\MM$. 
Such colimits exist because of the next statement, which follows 
from the similar result about limits and colimits 
in the categories of algebras over a colored operad (see e.g.\ 
\cite{Fresse_2009_ModulesOperads, PavlovScholbach_2018_AdmissibilityRectification}) 
by recalling that $\astAlg(\O)$ is just the category of algebras 
over the colored operad $\O$ in $\astObj(\MM)$, 
which has all limits and colimits that exist in the category $\MM$. 

\begin{propo}\label{propo:astAlg}
The category $\astAlg(\O)$ of $\ast$-algebras over a colored $\ast$-operad 
$\O \in \astOp$ is (co)complete if the underlying 
involutive symmetric monoidal category $\MM$ is complete 
(respectively closed and cocomplete). 
Furthermore, limits, filtered colimits and reflexive coequalizers 
are created in the category $\prod_{t \in \O} \astObj(\MM)$ of colored collections of $\ast$-objects, 
i.e.\ color-wise in $\astObj(\MM)$. 
\end{propo}

As already anticipated, we shall compute colimits in $\astAlg(\O)$, 
for a given colored $\ast$-operad $\O \in \astOp$ in $\MM$. 
Proposition~\ref{propo:astAlg} explains that this can be done 
color-wise in $\astObj(\MM)$ provided one 
manages to rephrase colimits into coequalizers 
that are color-wise reflexive in $\astObj(\MM)$. 
First, it is a general fact that colimits can be computed 
via coproducts and coequalizers. Furthermore, the coequalizer of 
the parallel pair $\Phi,\Psi\colon A_1 \to A_2$ in $\astAlg(\O)$ 
can be computed via the reflexive coequalizer 
\begin{equation}\label{eq:astAlg-coeq}
\colim \left( \xymatrix{A_1 \ar@<2pt>[r]^-{\Phi} \ar@<-2pt>[r]_-{\Psi} & A_2} \right) \coloneqq \colim \left( \xymatrix{A_1 \amalg A_2 \ar@<4pt>[r]^-{\langle \Phi,\id \rangle} \ar@<-4pt>[r]_-{\langle \Psi,\id \rangle} & A_2 \ar[l]|-{\iota_2}} \right) \in \astAlg(\O),
\end{equation}
where we use the notation $\langle f \rangle$ to denote the morphism 
that is defined by the list of morphisms $f = (f_i)$ 
on the respective components of the coproduct 
and $\iota_i$ denotes the morphism 
that includes the component labeled by $i$. 
Finally, the adjunction \eqref{eq:free-forget} allows us to compute the coproduct 
of $A_i \in \astAlg(\O)$, $i \in I$, via the reflexive coequalizer 
\begin{equation}\label{eq:astAlg-coprod}
\coprod_{i \in I} A_i \coloneqq \colim \left( \xymatrix{L \left( \displaystyle\coprod_{i \in I} R L R (A_i) \right) \ar@<4pt>[r]^-{d_0} \ar@<-4pt>[r]_-{d_1} & L \left( \displaystyle\coprod_{i \in I} R (A_i) \right) \ar[l]|-{s_0}} \right) \in \astAlg(\O),
\end{equation}
where $d_0 \coloneqq \varepsilon_{L(\coprod_i R(A_i))} \circ L(\langle R L (\iota_i), i \in I \rangle)$, 
$d_1 \coloneqq L(\coprod_i U(\varepsilon_{A_i}))$ and 
$s_0 \coloneqq L(\coprod_i \eta_{R(A_i)})$. 
(Here $\eta$ and $\varepsilon$ denote respectively the unit and 
the counit of the adjunction \eqref{eq:free-forget}.) 
We shall use these facts in Section~\ref{subsec:gluing} 
to present an explicit model for a colimit in the category of AQFTs, 
interpreted as algebras over the AQFT colored $\ast$-operad, 
as explained in Subsection \ref{subsec:AQFToperad}.

\subsection{\label{subsec:AQFToperad}Orthogonal categories and AQFT colored \texorpdfstring{$\ast$}{star}-operads}
This section recalls the concept of an orthogonal 
category and the assignment of the  associated AQFT 
colored $\ast$-operad. 
The functoriality of this assignment shall play 
a crucial role throughout the rest of the paper. 
More details on the topics discussed here can be found in 
\cite{BeniniSchenkelWoike_2019_InvolutiveCategories, BeniniSchenkelWoike_2021_OperadsAlgebraica, Yau_2020_InvolutiveMonoidal, BeniniGiorgettiSchenkel_2022_SkeletalModel, BeniniCarmonaSchenkel_2023_StrictificationTheorems}. 

\begin{defi}
An {\it orthogonal category} $\ovr{\CC} = (\CC,\perp)$ 
consists of a (small) category $\CC$ 
endowed with an {\it orthogonality relation} 
${\perp} \subseteq \Mor\,\CC \,{{}_{\mathsf{t}}{\times}{}_{\mathsf{t}}}\,\Mor\,\CC$, 
i.e.\ a subset of the set of pairs of morphisms to a common target 
$f_1\colon c_1 \to t \leftarrow c_2 \cocolon f_2$ in $\CC$ 
that is symmetric ($(f_1,f_2) \in {\perp}$ 
iff $(f_2,f_1) \in {\perp}$) and stable under composition 
(if $(f_1,f_2) \in {\perp}$, 
then $(h f_1 g_1,h f_2 g_2) \in {\perp}$ 
for all composable morphisms $g_1, g_2, h$ in $\CC$). 
Elements $(f_1,f_2) \in {\perp}$ are called {\it orthogonal pairs}. 
\sk

An {\it orthogonal functor} $F\colon \ovr{\CC} \to \ovr{\DD}$ 
is a functor $F\colon \CC \to \DD$ that is compatible 
with the orthogonality relations, 
i.e.\ $F(\perp_\CC) \subseteq {\perp_\DD}$. 
\sk

We denote by $\OCat$ the category whose objects are 
orthogonal categories $\ovr{\CC}$ and whose morphisms 
are orthogonal functors $F\colon \ovr{\CC} \to \ovr{\DD}$. 
\end{defi}

Several examples of orthogonal categories that feature 
in various flavours of algebraic quantum field theory, 
e.g.\ locally covariant quantum field theory 
\cite{BrunettiFredenhagenVerch_2003_GenerallyCovariant, FewsterVerch_2015_AlgebraicQuantum}, 
Haag--Kastler nets \cite{HaagKastler_1964_AlgebraicApproach}, 
as well as their conformal counterparts 
\cite{Pinamonti_2009_ConformalGenerally, CrawfordRejznerVicedo_2022_Lorentzian2D}, 
can be found in 
\cite{BeniniSchenkelWoike_2019_InvolutiveCategories, BeniniGiorgettiSchenkel_2022_SkeletalModel}. 
We focus on the examples that are of primary interest 
for the rest of this paper. 

\begin{ex}\label{ex:Man}
The first example of an orthogonal category that is relevant for 
this paper is $\ovr{\Man_m} = (\Man_m,\perp) \in \OCat$, 
with $m \geq 1$. 
Here $\Man_m$ denotes the category whose objects are 
$m$-dimensional oriented smooth manifolds and whose morphisms 
are orientation preserving open embeddings $f\colon M \to N$. 
Orthogonal pairs $(f_1,f_2) \in {\perp}$ are by definition pairs 
of morphisms $f_1\colon M_1 \to N \leftarrow M_2 \cocolon f_2$ in $\Man_m$ 
with disjoint images $f_1(M_1) \cap f_2(M_2) = \emptyset$. 
\end{ex}

\begin{ex}\label{ex:slice}
Given an orthogonal category $\ovr{\CC} = (\CC,\perp) \in \OCat$ 
and an object $t \in \CC$, one can form the {\it slice} 
orthogonal category $\ovr{\CC_{/t}} \in \OCat$ 
by endowing the slice category $\CC_{/t}$ with the orthogonality 
relation given by the pullback of $\perp$ 
along the forgetful functor $\CC_{/t} \to \CC$, 
namely a pair 
\begin{equation}
\xymatrix{
f_1\colon c_1 \ar[dr]_-{g_1} \ar[r] & c \ar[d]_-{h} & c_2 \cocolon f_2 \ar[dl]^-{g_2} \ar[l] \\
& t
}
\end{equation}
in $\CC_{/t}$ is orthogonal iff $(f_1,f_2) \in {\perp}$. 
In particular, the forgetful functor becomes an orthogonal functor 
$\ovr{\CC_{/t}} \to \ovr{\CC}$. 
\sk

The construction of the slice orthogonal category 
applied to $\ovr{\Man_m}$ provides 
the orthogonal categories most relevant for the rest 
of the paper, namely for all $M \in \Man_m$, we shall consider 
the orthogonal categories $\ovr{\Man_{m\,/M}} \in \OCat$.

Let us observe that these slice orthogonal categories 
admit an equivalent, yet much more concrete, 
interpretation in terms of open subsets. 
Indeed, consider the orthogonal category 
$\ovr{\OO(M)} = (\OO(M),\perp) \in \OCat$ consisting of 
the category $\OO(M)$, whose objects are 
non-empty open subsets $U \subseteq M$ and whose morphisms are 
inclusions $U \subseteq V$, endowed with 
the orthogonality relation $\perp$ formed by all pairs of morphism with common target $(U_1 \subseteq V,\,U_2 \subseteq V) \in {\perp}$ 
in $\OO(M)$ that are disjoint, namely $U_1 \cap U_2 = \emptyset$. 
One realizes immediately that $\ovr{\OO(M)}$ 
is orthogonally equivalent to the slice orthogonal category 
$\ovr{\Man_{m\,/M}}$. In particular, there is 
an orthogonal functor $\iota_M\colon \ovr{\OO(M)} \to \ovr{\Man_m}$ 
in $\OCat$ that sends an open subset $U \subseteq M$ to $U \in \Man_m$ by endowing 
the open subset $U \subseteq M$ with the restriction 
of the orientation and smooth manifold structure of $M$. 
\end{ex}

We now review the definition of the AQFT colored $\ast$-operad 
associated with an orthogonal category. For more details see 
\cite{BeniniSchenkelWoike_2019_InvolutiveCategories, BeniniSchenkelWoike_2021_OperadsAlgebraica, BeniniCarmonaSchenkel_2023_StrictificationTheorems}.  

\begin{defi}\label{defi:AQFToperad}
Let $\MM$ be an involutive symmetric monoidal category. 
The {\it AQFT colored $\ast$-operad} $\O_{\ovr{\CC}} \in \astOp$ 
in $\MM$ associated with a (small) orthogonal category 
$\ovr{\CC} \in \OCat$ consists of the data listed below: 
\begin{enumerate}[label=(\alph*)]
\item The set of colors of $\O_{\ovr{\CC}}$ coincides with 
the set of objects of $\ovr{\CC}$. 
\item For all $t \in \O_{\ovr{\CC}}$ and all tuples 
$\und{c} = (c_1, \ldots, c_n)$ in $\O_{\ovr{\CC}}$ of length $n \geq 0$, 
the object of {\it operations} 
\begin{equation}
\O_{\ovr{\CC}}(\substack{t\\ \und{c}}) \coloneqq \big( \Sigma_n \times \CC(\und{c},t) \big) \big/ \sim_{\perp} \otimes \oone \in \MM,
\end{equation}
where $\otimes$ denotes the $\Set$-tensoring on $\MM$, 
$\CC(\und{c},t) \coloneqq \prod_{i=1}^n \CC(c_i,t)$ is the set 
of $n$-tuples of morphisms in $\CC$ with sources 
$\und{c} = (c_1,\ldots,c_n)$ and common target $t$ 
and $\sim_{\perp}$ is the equivalence relation defined as follows: 
$(\tau,\und{g}) \sim_{\perp} (\tau,\und{g}^\prime)$ 
if and only if $\und{g} = \und{g}^\prime$ and the right permutation 
$\tau \tau^{\prime\,-1}\colon \und{g} \tau^{-1} \to \und{g} \tau^{\prime\,-1}$ 
is generated by transpositions of adjacent orthogonal pairs. 
We denote the equivalence class of $(\tau,\und{g})$ by $[\tau,\und{g}]$.
\item For all $t \in \O_{\ovr{\CC}}$, all tuples $\und{b}$ in $\O_{\ovr{\CC}}$ 
of length $m \geq 1$ and all tuples $\und{a}_i$ in $\O_{\ovr{\CC}}$ 
of length $k_i \geq 0$, $i=1,\ldots,m$, the composition morphism 
\begin{subequations}
\begin{equation}
\gamma\colon \O_{\ovr{\CC}}(\substack{t\\ \und{b}}) \otimes \bigotimes_{i=1}^m \O_{\ovr{\CC}}(\substack{b_i\\ \und{a}_i}) \longrightarrow \O_{\ovr{\CC}}(\substack{t\\ \und{\und{a}}})
\end{equation}
in $\MM$ defined by the $\Set$-tensoring $(-) \otimes \oone$ of the map 
\begin{align}
\big( \Sigma_m \times \CC(\und{b},t) \big) \big/ \sim_{\perp} \times \prod_{i=1}^m \big( \Sigma_{k_i} \times \CC(\und{a}_i,b_i) \big) \big/ \sim_{\perp} &\longrightarrow \big( \Sigma_n \times \CC(\und{\und{a}},t) \big) \big/ \sim_{\perp} \\ 
\big( [\tau,\und{g}], [\sigma_1,\und{f}_1], \ldots, [\sigma_m,\und{f}_m] \big) &\longmapsto \big[ \tau(\sigma_1,\ldots,\sigma_m),\und{g}(\und{f}_1,\ldots,\und{f}_m) \big], \nn 
\end{align}
\end{subequations}
where $n \coloneqq \sum_i k_i$ is the length of $\und{\und{a}}$, 
$\tau(\sigma_1,\ldots,\sigma_m) \in \Sigma_n$ 
is the composition of the block permutation induced by $\tau$ 
and the block sum permutation $(\sigma_1,\ldots,\sigma_m)$ 
and $\und{g}(\und{f}_1,\ldots,\und{f}_m) \coloneqq (g_1 f_{1 1}, \ldots, g_m f_{m k_m}) \in \CC(\und{\und{a}},t)$ 
is the $n$-tuple given by composition in $\CC$. 
\item For all $t \in \O_{\ovr{\CC}}$, the identity morphism 
\begin{subequations}
\begin{equation}
\id_t\colon \oone \longrightarrow \O_{\ovr{\CC}}(\substack{t\\ t})
\end{equation}
in $\MM$ defined by the $\Set$-tensoring $(-) \otimes \oone$ of the map 
\begin{align}
\{\ast\} &\longrightarrow \big( \Sigma_1 \times \CC(t,t) \big) \big/ \sim_{\perp} \\ 
\ast &\longmapsto [e,\id_t], \nn 
\end{align}
\end{subequations}
where $e \in \Sigma_1$ is the identity permutation 
and $\id_t \in \CC(t,t)$ is the identity morphism. 
\item For all $t \in \O_{\ovr{\CC}}$, all tuples $\und{c}$ in $\O_{\ovr{\CC}}$ 
of length $n \geq 0$ and all permutations $\sigma \in \Sigma_n$, 
the morphism 
\begin{subequations}
\begin{equation}
\O_{\ovr{\CC}}(\sigma)\colon \O_{\ovr{\CC}}(\substack{t\\ \und{c}}) \longrightarrow \O_{\ovr{\CC}}(\substack{t\\ \und{c}\sigma})
\end{equation}
in $\MM$ defined by the $\Set$-tensoring $(-) \otimes \oone$ of the map 
\begin{align}
\big( \Sigma_n \times \CC(\und{c},t) \big) \big/ \sim_{\perp} &\longrightarrow \big( \Sigma_n \times \CC(\und{c}\sigma,t) \big) \big/ \sim_{\perp} \\ 
[\tau,\und{g}] &\longmapsto [\tau \sigma, \und{g} \sigma]. \nn 
\end{align}
\end{subequations}
\item For all $t \in \O_{\ovr{\CC}}$ and all tuples $\und{c}$ in $\O_{\ovr{\CC}}$ 
of length $n \geq 0$, the {\it $\ast$-involution} 
\begin{subequations}
\begin{equation}
\ast\colon \O_{\ovr{\CC}}(\substack{t\\ \und{c}}) \longrightarrow J(\O_{\ovr{\CC}}(\substack{t\\ \und{c}}))
\end{equation}
in $\MM$ defined by the $\Set$-tensoring $(-) \otimes J_0$ 
(see Definition~\ref{defi:ISMCat} for the structure morphism 
$J_0\colon \oone \to J(\oone)$ in $\MM$) of the map 
\begin{align}
\big( \Sigma_n \times \CC(\und{c},t) \big) \big/ \sim_{\perp} &\longrightarrow \big( \Sigma_n \times \CC(\und{c},t) \big) \big/ \sim_{\perp} \\ 
[\tau,\und{g}] \mapsto [\rho_n \tau, \und{g}], \nn 
\end{align}
\end{subequations}
where $\rho_n \in \Sigma_n$ denotes the order reversing permutation. 
\end{enumerate}
These data fulfil the axioms of a colored $\ast$-operad, 
see \cite{BeniniSchenkelWoike_2019_InvolutiveCategories}. 
\end{defi}

\begin{defi}\label{defi:AQFTmultifunctor}
Let $\MM$ be a cocomplete involutive symmetric monoidal category. 
The {\it AQFT $\ast$-multifunctor} 
$F\colon \O_{\ovr{\CC}} \to \O_{\ovr{\DD}}$ in $\astOp$ in $\MM$ 
associated with an orthogonal functor $F\colon \ovr{\CC} \to \ovr{\DD}$ 
in $\OCat$ consists of the data listed below: 
\begin{enumerate}[label=(\alph*)]
\item The map between the underlying sets of colors coincides 
with the map between sets of objects underlying $F\colon \ovr{\CC} \to \ovr{\DD}$. 
\item For all $t \in \O_{\ovr{\CC}}$ and all tuples $\und{c}$ in $\O_{\ovr{\CC}}$ 
of length $n \geq 0$, the morphism 
\begin{subequations}
\begin{equation}
F\colon \O_{\ovr{\CC}}(\substack{t\\ \und{c}}) \longrightarrow \O_{\ovr{\DD}}(\substack{Ft \\ F\und{c}})
\end{equation}
in $\MM$ defined by the $\Set$-tensoring $(-) \otimes \oone$ of the map 
\begin{align}
\big( \Sigma_n \times \CC(\und{c},t) \big) \big/ \sim_{\perp} &\longrightarrow \big( \Sigma_n \times \DD(F\und{c},Ft) \big) \big/ \sim_{\perp} \\ 
[\tau,\und{g}] &\longmapsto \big[ \tau,F\und{g} \big], \nn 
\end{align}
\end{subequations}
where $F\und{g} \coloneqq (Fg_1,\ldots,Fg_n)$ is the $n$-tuple 
given by the action of $F\colon \ovr{\CC} \to \ovr{\DD}$ 
on each component of $\und{g}$. 
\end{enumerate}
\end{defi}

\begin{defi}\label{defi:AQFToperad-functor}
The {\it AQFT operad} functor 
\begin{equation}
\O_{(-)}\colon \OCat \longrightarrow \astOp
\end{equation}
sends an orthogonal category $\ovr{\CC} \in \OCat$ 
to the AQFT colored $\ast$-operad $\O_{\ovr{\CC}} \in \astOp$, 
see Definition~\ref{defi:AQFToperad}, 
and an orthogonal functor $F\colon \ovr{\CC} \to \ovr{\DD}$ in $\OCat$ 
to the $\ast$-multifunctor $F\colon \O_{\ovr{\CC}} \to \O_{\ovr{\DD}}$ 
in $\astOp$, see Definition~\ref{defi:AQFTmultifunctor}. 
\end{defi}

The AQFT colored $\ast$-operad $\O_{\ovr{\CC}}$ allows for a concise 
definition of AQFTs over $\ovr{\CC}$. 

\begin{defi}\label{defi:AQFT}
Let $\ovr{\CC} = (\CC,\perp) \in \OCat$ be an orthogonal category 
and $\MM$ an involutive symmetric monoidal category. 
The category $\AQFT(\ovr{\CC}) \coloneqq \astAlg(\O_{\ovr{\CC}})$ 
of {\it $\MM$-valued AQFTs over $\ovr{\CC}$} is defined as that 
of $\ast$-algebras over the AQFT colored $\ast$-operad 
$\O_{\ovr{\CC}} \in \astOp$, see Definitions~\ref{defi:astAlg} 
and~\ref{defi:AQFToperad}. 
\end{defi}

\begin{rem}\label{rem:AQFTsAsFunctors}
The previous definition is extremely concise, however quite 
abstract. One might ask what is an AQFT 
$\AAA \in \AQFT(\ovr{\CC})$ in more concrete terms. 
It is shown in \cite{BeniniSchenkelWoike_2019_InvolutiveCategories, BeniniSchenkelWoike_2021_OperadsAlgebraica} 
that $\AAA$ admits the following more familiar description: 
an AQFT $\AAA \in \AQFT(\ovr{\CC})$ is a functor 
$\AAA\colon \CC \to \ast\Monr(\MM)$ from the category $\CC$ 
to the category of 
{\it reversing} $\ast$-monoids in $\MM$ such that, 
for all orthogonal pairs $f_1\colon c_1 \to t \leftarrow c_2 \cocolon f_2$, 
the diagram 
\begin{equation}
\xymatrix@C=5em{
\AAA(c_1) \otimes \AAA(c_2) \ar[r]^-{\AAA(f_1) \otimes \AAA(f_2)} \ar[d]_-{\AAA(f_1) \otimes \AAA(f_2)} & \AAA(t) \otimes \AAA(t) \ar[d]^-{\mu} \\
\AAA(t) \otimes \AAA(t) \ar[r]_-{\mu^{\op}} & \AAA(t) 
}
\end{equation}
in $\MM$ commutes, where $\mu^{(\op)}$ denotes the (opposite) 
multiplication of $\AAA(t) \in \ast\Monr(\MM)$. 
This is (a generalized version of) the well-known causality axiom. 
\sk

In the same fashion, it turns out that a morphism 
$\Phi\colon \AAA \to \BBB$ in $\AQFT(\ovr{\CC})$ is precisely 
a natural transformation between the underlying functors 
$\AAA, \BBB\colon \CC \to \ast\Monr(\MM)$. 
\end{rem}


\section{\label{sec:fibcat}Gluing algebraic quantum field theories}
Let $\MM$ be an involutive symmetric monoidal category 
and consider the orthogonal category $\ovr{\Man_m}$ 
of $m$-dimensional oriented smooth manifolds. 
In this section we explain first that $\MM$-valued AQFTs 
over the orthogonal category $\ovr{\OO(M)}$ of open subsets, 
for all $M \in \ovr{\Man_m}$, 
form a category fibred over $\Man_m$. Next, we construct 
an explicit and convenient functor that extends 
AQFTs over $\ovr{\OO(M)}$ to AQFTs over $\ovr{\OO(N)}$ 
along a morphism $f\colon M \to N$ in $\Man_m$. 
This functor realizes a specific instance of 
the left adjoint from \eqref{eq:F!F*}. Finally, 
we explain how to glue AQFTs over a cover of $M \in \ovr{\Man_m}$ 
by computing a suitable colimit, of which we provide 
an explicit model that shall play 
a crucial role in Section~\ref{sec:descent}.

\subsection{\label{subsec:fib-cat}The fibred category of AQFTs}
Consider the orthogonal category $\ovr{\Man_m}$ 
of $m$-dimensional smooth oriented manifolds from Example~\ref{ex:Man} 
and an involutive symmetric monoidal category $\MM$. 
For each $M \in \Man_m$, recall the orthogonal category 
$\ovr{\OO(M)}$ of open subsets of $M$ from Example~\ref{ex:slice}. 
Note that each morphism $f\colon M \to N$ in $\Man_m$ yields 
an orthogonal functor 
\begin{equation}
f\colon \ovr{\OO(M)} \longrightarrow \ovr{\OO(N)} 
\end{equation}
that sends an open subset 
$U \subseteq M$ to its image $f(U) \subseteq N$ 
(which is an open subset because $f$ is an open embedding). 
Combining Definitions~\ref{defi:AQFToperad-functor} 
and~\ref{defi:AQFT} with \eqref{eq:F!F*}, the orthogonal 
functor $f$ induces the extension-restriction adjunction 
\begin{equation}\label{eq:ext-res-Man}
\xymatrix{
f_! \colon \AQFT(\ovr{\OO(M)}) \ar@<2pt>[r] & \AQFT(\ovr{\OO(N)}) \cocolon f^\ast \ar@<2pt>[l]
}.
\end{equation}

Regarding AQFTs as functors that assign 
reversing $\ast$-monoids, see Remark~\ref{rem:AQFTsAsFunctors}, 
the extension-restriction adjunction $f_! \dashv f^\ast$ admits 
a very explicit description. 

The right adjoint functor $f^\ast$, called the restriction along $f$, 
sends an AQFT $\BBB \in \AQFT(\ovr{\OO(N)})$ 
to the AQFT $f^\ast \BBB \in \AQFT(\ovr{\OO(M)})$ 
that assigns to each open subset $U \subseteq M$ the reversing 
$\ast$-monoid $f^\ast \BBB(U) \coloneqq \BBB(f(U)) \in \ast\Monr(\MM)$. 

The left adjoint functor $f_!$, called the extension along $f$, sends an AQFT $\AAA \in \AQFT(\ovr{\OO(M)})$ 
to the AQFT $f_! \AAA \in \AQFT(\ovr{\OO(N)})$ 
that assigns to each open subset $V \subseteq N$ the reversing 
$\ast$-monoid $f_! \AAA(V) \in \ast\Monr(\MM)$ defined as follows: 
when $f(M) \cap V = \emptyset$, i.e.\ the image of $f$ does not meet $V$, 
then $f_! \AAA(V) \in \ast\Monr(\MM)$ is the initial 
reversing $\ast$-monoid; when $f(M) \cap V \neq \emptyset$, 
i.e.\ the image of $f$ meets $V$, 
$f_! \AAA(V) \coloneqq \AAA(f^{-1}(V)) \in \ast\Monr(\MM)$ 
is the reversing $\ast$-monoid that $\AAA$ assigns to the 
non-empty open subset $f^{-1}(V) \subseteq M$. 

The unit of the extension-restriction adjunction is given component-wise, 
for all $\AAA \in \linebreak \AQFT(\ovr{\OO(M)})$, 
as the natural transformation $\AAA \to f^\ast f_! \AAA$ 
between functors from $\OO(M)$ to $\ast\Monr(\MM)$ 
whose component at the open subset $U \subseteq M$ 
is just the identity of $\AAA(U) \in \ast\Monr(\MM)$. 

The counit of the extension-restriction adjunction is given component-wise, 
for all $\BBB \in \AQFT(\ovr{\OO(N)})$, 
as the natural transformation $f_! f^\ast \BBB \to \BBB$ 
between functors from $\OO(N)$ to $\ast\Monr(\MM)$ 
whose component at the open subset $V \subseteq N$ 
is the morphism $f_! f^\ast \BBB (V) \to \BBB(V)$ in $\ast\Monr(\MM)$ 
defined as follows: when $f(M) \cap V = \emptyset$ is empty, 
it is given by the unique morphism 
from the intial reversing $\ast$-monoid to $\BBB(V)$; 
when $f(M) \cap V \neq \emptyset$ is non-empty, it is the morphism 
$\BBB(\subseteq)\colon \BBB(f(f^{-1}(V))) \to \BBB(V)$ in $\ast\Monr(\MM)$ 
that the $\BBB$ assigns 
to the non-empty open subset inclusion 
$f(f^{-1}(V)) \subseteq V$. 

\begin{defi}\label{defi:pi}
Let $\MM$ be an involutive symmetric monoidal category. 
We define the category $\AQFT$ whose objects are pairs $(M,\AAA)$ 
consisting of $M \in \Man_m$ and $\AAA \in \AQFT(\ovr{\OO(M)})$, 
and whose morphisms $(f,\Phi)\colon (M,\AAA) \to (N,\BBB)$ are pairs 
consisting of a morphism $f\colon M \to N$ in $\Man_m$ and a morphism 
$\Phi\colon \AAA \to f^\ast \BBB$ in $\AQFT(\ovr{\OO(M)})$. 
\sk

We define the functor $\pi\colon \AQFT \to \Man_m$ that sends objects 
$(M,\AAA) \in \AQFT$ to $M \in \Man_m$ and morphisms 
$(f,\Phi)\colon (M,\AAA) \to (N,\BBB)$ in $\AQFT$ to $f\colon M \to N$ 
in $\Man_m$. 
\end{defi}

The following simple result explains the title of this section. 

\begin{propo}
The functor $\pi\colon \AQFT \to \Man_m$ from Definition~\ref{defi:pi} 
is a fibred category, that we call the {\it AQFT fibred category}. 
\end{propo}
\begin{proof}
The goal is to show that, for all $g\colon M \to N$ in $\Man_m$ and 
$(N,\BBB) \in \AQFT$, the morphism 
$(g,\id)\colon (M,g^\ast \BBB) \to (N,\BBB)$ in $\AQFT$ is Cartesian. 
Indeed, for all $f:L \to M$ in $\Man_m$ and 
all $(g f,\Phi)\colon (L,\AAA) \to (N,\BBB)$ in $\AQFT$, 
the unique morphism 
$\chi\colon (L,\AAA) \to (M,g^\ast \BBB)$ in $\AQFT$ 
such that $\pi(\chi) = f$ and $(g,\id) \circ \chi = (g f,\Phi)$, 
is given by $\chi \coloneqq (f,\Phi)$. 
\end{proof}

\subsection{\label{subsec:gluing}Gluing AQFTs on an open cover}
Section~\ref{subsec:fib-cat} defines the fibred category 
$\pi\colon \AQFT \to \Man_m$, whose fiber over $M \in \Man_m$ 
is the category $\pi^{-1}(M) = \AQFT(\ovr{\OO(M)})$ 
of AQFTs on the orthogonal category $\ovr{\OO(M)}$ 
of open subsets of $M$, see Example~\ref{ex:slice}. 
Furthermore, \eqref{eq:ext-res-Man} 
presents an explicit extension-restriction adjunction. 

Consider now also an open cover 
$\MMM = \{M_\alpha\}$ of an object $M \in \Man_m$. 
The goal of this section is to present a canonical construction that {\it glues} 
{\it local} AQFTs $\AAA_\alpha \in \AQFT(\ovr{\OO(M_\alpha)})$ 
on $\ovr{\OO(M_\alpha)}$ (i.e.\ defined on the covering patches
$M_\alpha$) that agree on overlaps to a {\it global} AQFT 
$\glue(\{\AAA_\alpha\}) \in \AQFT(\ovr{\OO(M)})$ on $\ovr{\OO(M)}$ 
(i.e.\ on the whole of $M$). 

Note that this construction 
does {\it not} merely glue, for each open subsets $U \subseteq M$, 
the local algebras $\AAA_{\alpha}(U \cap M_\alpha)$ 
(to be interpreted as the 
initial reversing $\ast$-monoid when $U$ does not meet $M_\alpha$) 
to a global algebra. 
Instead, it glues the local AQFTs to a global AQFT and, 
in doing so, our gluing construction implements the AQFT axioms, 
which is in general not the case when gluing 
the local algebras. 

Furthermore, as we will see in Section~\ref{sec:descent}, 
our gluing construction solves the issue of not being able to 
recover the global data from the local data associated with an AQFT 
even for an extremely simple example of a $1$-dimensional AQFT, as
already mentioned in the introduction. 

Our gluing construction consists of an adjunction  
\begin{equation}\label{eq:glue-data}
\xymatrix{
\glue\colon \Desc(\MMM) \ar@<2pt>[r] & \AQFT(\ovr{\OO(M)}) \cocolon \data \ar@<2pt>[l]
}. 
\end{equation}
defined on a suitable category $\Desc(\MMM)$ of 
{\it AQFT descent data} subordinate to the open cover $\MMM$ of $M$, 
consisting of local AQFTs that agree on overlaps. 
This is formalized more accurately in Definition~\ref{defi:Desc} below. 
\sk

To define the category $\Desc(\MMM)$ of AQFT descent data 
subordinate to an open cover $\MMM = \{M_\alpha\}$ of $M$, 
we specialize the concept of category of descent data 
of \cite[Sec.~4.1]{Vistoli_2005_GrothendieckTopologies} 
to the fibred category $\pi\colon \AQFT \to \Man_m$. 
For this purpose let 
$(-)\vert^{U^\prime} \dashv (-)\vert_U\colon \AQFT(\ovr{\OO(U^\prime)}) \to \AQFT(\ovr{\OO(U)})$ 
denote the extension-restriction adjunction \eqref{eq:ext-res-Man}  
associated with the $\Man_m$-morphism given by an inclusion of 
non-empty open subsets $U \subseteq U^\prime$ of $M$.  
Furthermore, denote $n$-fold overlaps of the cover $\MMM$ by 
$M_{\alpha_1 \ldots \alpha_n} \coloneqq \bigcap_{i=1}^n M_{\alpha_i}$. 

\begin{conv}
Note that an $n$-fold overlap $M_{\alpha_1 \ldots \alpha_n}$ 
of a cover $\MMM$ may be empty. In this case, it should be omitted 
from the constructions below. Removing empty overlaps, however, 
leads to a quite clumsy notation. Instead, we adopt the following 
more convenient convention. 
When $\emptyset \subseteq M$ is the empty open subset, 
$\AQFT(\ovr{\O(\emptyset)}) \coloneqq \{\ast\}$ is the terminal category 
(consisting of a single object $\ast$ 
and its identity morphism $\id_\ast$). 
When $\emptyset \subseteq U^\prime$ is the inclusion of 
the empty open subset inside an open subset $U^\prime \subseteq M$, 
the restriction $(-)\vert_{\emptyset}\colon \AQFT(\ovr{\OO(U^\prime)}) \to \AQFT(\ovr{\OO(\emptyset)})$ 
is the unique functor to the terminal category 
(that sends all objects to $\ast$ and all morphisms to $\id_\ast$), 
and the extension $(-)\vert^{U^\prime}\colon \AQFT(\ovr{\OO(\emptyset)}) \to \AQFT(\ovr{\OO(U^\prime)})$ 
is the functor that assigns the initial AQFT 
(i.e.\ the AQFT that constantly assigns the initial reversing $\ast$-monoid) 
if $U^\prime \neq \emptyset$ is a non-empty open subset and 
the object $\ast$ if $U^\prime = \emptyset$ 
is the empty open subset. 
\end{conv}

\begin{defi}\label{defi:Desc}
The category $\Desc(\MMM)$ of AQFT descent data on a 
cover $\MMM = \{M_\alpha\}$ of $M$ is defined as follows. 
\begin{enumerate}[label=(\alph*)]
\item Objects $\AAAA = (\{\AAA_\alpha\},\{\aaa_{\alpha\beta}\})$ 
consist of a collection of AQFTs 
\begin{subequations}
\begin{equation}
\AAA_\alpha \in \AQFT(\ovr{\OO(M_\alpha)})
\end{equation}
on $\ovr{\OO(M_\alpha)}$ on the patches $M_\alpha$ for all $\alpha$, 
and a collection of isomorphisms 
\begin{equation}\label{eq:overlap-isos}
\aaa_{\alpha\beta}\colon \AAA_\alpha \vert_{M_{\alpha\beta}} \overset{\cong}{\longrightarrow} \AAA_\beta \vert_{M_{\alpha\beta}}
\end{equation}
in $\AQFT(\ovr{\OO(M_{\alpha\beta}})$ on the double overlaps 
$M_{\alpha\beta}$ for all $\alpha,\beta$, 
subject to the conditions 
\begin{align}
\aaa_{\alpha\alpha} = \id, && \aaa_{\beta\gamma} \vert_{M_{\alpha\beta\gamma}} \circ \aaa_{\alpha\beta} \vert_{M_{\alpha\beta\gamma}} = \aaa_{\alpha\gamma} \vert_{M_{\alpha\beta\gamma}}
\end{align}
\end{subequations}
for all $\alpha,\beta,\gamma$. 
\item Morphisms $\PPhi = \{\Phi_\alpha\}\colon \AAAA \to \BBBB$ 
consist of a collection of morphisms 
\begin{subequations}
\begin{equation}
\Phi_\alpha\colon \AAA_\alpha \longrightarrow \BBB_\alpha
\end{equation}
in $\AQFT(\ovr{\OO(M_\alpha)})$ on the patches $M_\alpha$, 
for all $\alpha$, subject to the conditions 
\begin{equation}
\bbb_{\alpha\beta} \circ \Phi_\alpha \vert_{M_{\alpha\beta}} = \Phi_\beta \vert_{M_{\alpha\beta}} \circ \aaa_{\alpha\beta}, 
\end{equation}
\end{subequations}
for all $\alpha,\beta$. 
\end{enumerate}
Together with the evident identity morphisms and composition law, 
the objects and morphisms defined above form the category $\Desc(\MMM)$. 
\end{defi}

The candidate right adjoint functor 
\begin{subequations}\label{eq:data}
\begin{equation}
\data\colon \AQFT(\ovr{\OO(M)}) \longrightarrow \Desc(\MMM)
\end{equation}
from \eqref{eq:glue-data} assigns to an AQFT 
$\AAA \in \AQFT(\ovr{\OO(M)})$ the descent datum 
\begin{equation}
\data(\AAA) \coloneqq (\{\AAA \vert_{M_\alpha}\},\{\id\}) \in \Desc(\MMM) 
\end{equation} 
and to a morphism $\Phi\colon \AAA \to \BBB$ in $\AQFT(\ovr{\OO(M)})$ 
the morphism 
\begin{equation}
\data(\Phi) \coloneqq \{\Phi \vert_{M_\alpha}\}\colon \data(\AAA) \longrightarrow \data(\BBB)
\end{equation}
\end{subequations}
in $\Desc(\MMM)$. 
\sk

In order to define also the candidate left adjoint functor 
\begin{subequations}\label{eq:glue}
\begin{equation}
\glue\colon \Desc(\MMM) \longrightarrow \AQFT(\ovr{\OO(M)}), 
\end{equation}
recall that the category $\AQFT(\ovr{\OO(M)})$ is cocomplete 
(see Proposition~\ref{propo:astAlg}, Example~\ref{ex:slice} 
and Definition~\ref{defi:AQFT}) and consider 
the extension-restriction adjunction 
$(-)\vert^{U^\prime} \dashv (-)\vert_U\colon \AQFT(\ovr{\OO(U^\prime)}) \to \AQFT(\ovr{\OO(U)})$  
associated with the inclusion of 
open subsets $U \subseteq U^\prime$ of $M$. 

With these preparations, the action of $\glue$ on a descent datum 
$\AAAA = (\{\AAA_\alpha\},\{\aaa_{\alpha\beta}\})$ 
is defined as the reflexive coequalizer 
\begin{equation}\label{eq:glue-formula}
\glue(\AAAA) \coloneqq \colim \left( \xymatrix{ \coprod_{\alpha\beta} \AAA_{\alpha} \vert_{M_{\alpha\beta}} \vert^M \ar@<4pt>[r] \ar@<-4pt>[r] & \coprod_\alpha \AAA_\alpha \vert^M \ar[l] } \right) \in \AQFT(\ovr{\OO(M)}), 
\end{equation}
where the top arrow is defined from the component 
$\AAA_{\alpha} \vert_{M_{\alpha\beta}} \vert^{M_\alpha} \to \AAA_{\alpha}$ in $\AQFT(\ovr{\OO(M_\alpha)})$
of the counit of the extension-restriction adjunction, 
the bottom arrow is defined combining the isomorphism 
$\aaa_{\alpha\beta}\colon \AAA_\alpha \vert_{M_{\alpha\beta}} \to \AAA_{\beta} \vert_{M_{\alpha\beta}}$ in $\AQFT(\ovr{\OO(M_{\alpha\beta})})$
and the component
$\AAA_{\beta} \vert_{M_{\alpha\beta}} \vert^{M_\beta} \to \AAA_{\beta}$ in $\AQFT(\ovr{\OO(M_\beta)})$ 
of the counit of the extension-restriction adjunction, and 
the middle arrow is defined by the identity 
$\AAA_\alpha = \AAA_\alpha \vert_{M_{\alpha\alpha}}$. 
\end{subequations}

The action of $\glue$ on a morphism 
$\PPhi = \{\Phi_\alpha\}\colon \AAAA \to \BBBB$ in $\Desc(\MMM)$ 
is obtained combining the functoriality of coequalizers and 
the commutative diagram 
\begin{equation}
\xymatrix{
\coprod_{\alpha\beta} \AAA_{\alpha} \vert_{M_{\alpha\beta}} \vert^M \ar@<4pt>[r] \ar@<-4pt>[r] \ar[d]_-{\coprod_{\alpha\beta} \Phi_\alpha \vert_{M_{\alpha\beta}} \vert^M} & \coprod_\alpha \AAA_\alpha \vert^M \ar[d]^-{\coprod_{\alpha} \Phi_\alpha \vert^M} \ar[l] \\ 
\coprod_{\alpha\beta} \BBB_{\alpha} \vert_{M_{\alpha\beta}} \vert^M \ar@<4pt>[r] \ar@<-4pt>[r] & \coprod_\alpha \BBB_\alpha \vert^M \ar[l]
}
\end{equation}
in $\AQFT(\ovr{\OO(M)})$. The square involving 
the top arrows commutes as a consequence of the naturality 
of the extension-restriction adjunction counit. 
Commutativity of the square involving the bottom arrows 
relies also on the compatibility between $\PPhi$ and the 
collections $\{\aaa_{\alpha\beta}\}$, $\{\bbb_{\alpha\beta}\}$, 
see Definition~\ref{defi:Desc}. 
The square involving the middle arrows commutes due to 
the identity $\Phi_\alpha\vert_{M_{\alpha\alpha}} = \Phi_\alpha$. 
\sk 

The unit 
\begin{subequations}\label{eq:glue-data-unit}
\begin{equation}
\id \longrightarrow \data \circ \glue 
\end{equation}
of the adjunction \eqref{eq:glue-data} consists, for each 
$\AAAA = \{\AAA_\alpha,\aaa_{\alpha\beta}\} \in \Desc(\MMM)$, 
of the component 
\begin{equation}
\AAAA \longrightarrow \data(\glue(\AAAA)) 
\end{equation}
in $\Desc(\MMM)$ defined component-wise, 
for each $\alpha$, as the composite morphism 
\begin{equation}
\AAA_\alpha \longrightarrow \AAA_\alpha \vert^M \vert_{M_\alpha} \longrightarrow \glue(\AAAA) \vert_{M_\alpha} 
\end{equation}
\end{subequations}
in $\AQFT(\ovr{\OO(M_\alpha)})$, where the first morphism is the 
$\AAA_\alpha$-component of the unit of the extension-restriction 
adjunction $(-)\vert^{M} \dashv (-)\vert_{M_\alpha}$ 
and the second morphism is the restriction $(-)\vert_{M_\alpha}$ 
of the canonical morphism that includes $\AAA_\alpha\vert^M$ 
in the coequalizer $\glue(\AAA)$, see \eqref{eq:glue}. 
\sk 

Finally, the counit 
\begin{subequations}\label{eq:glue-data-counit}
\begin{equation}
\glue \circ \data \longrightarrow \id
\end{equation}
of the adjunction \eqref{eq:glue-data} consists, for each 
$\AAA \in \AQFT(\ovr{\OO(M)})$, of the component
\begin{equation}
\glue(\data(\AAA)) \longrightarrow \AAA 
\end{equation} 
in $\AQFT(\ovr{\OO(M)})$, which is defined 
as the canonical morphism arising via the universal property 
of the coequalizer from the diagram 
\begin{equation}
\xymatrix{
\coprod_{\alpha\beta} \AAA \vert_{M_{\alpha\beta}} \vert^M \ar@<2pt>[r] \ar@<-2pt>[r] & \coprod_\alpha \AAA \vert_{M_\alpha} \vert^M \ar[r] & \AAA
}
\end{equation}
\end{subequations}
in $\AQFT(\ovr{\OO(M)})$, where the morphism on the right-hand side 
is given component-wise, for each $\alpha$, by the counit 
of the extension-restriction adjunction 
$(-)\vert^{M} \dashv (-)\vert_{M_\alpha}$ and the parallel pair 
on the left-hand side is defined as in \eqref{eq:glue}. 
(To check that the top and bottom compositions coincide, recall 
that the coherence morphisms of $\data(\AAA) \in \Desc(\MMM)$ 
are identities by definition and that restrictions compose strictly.) 
\sk 

From the triangle identities of the extension-restriction 
adjunction it follows that the functor $\data$ in \eqref{eq:data}, 
the functor $\glue$ in \eqref{eq:glue}, 
the unit \eqref{eq:glue-data-unit} and the counit 
\eqref{eq:glue-data-counit} form an adjunction, 
as claimed in \eqref{eq:glue-data}. 

\begin{rem}\label{rem:glue-data-counit}
Incidentally, let us observe that the unit 
\eqref{eq:glue-data-unit} of the adjunction $\glue \dashv \data$ 
in \eqref{eq:glue-data} is a natural isomorphism, 
therefore one can identify the category of descent data 
$\Desc(\MMM)$ as a coreflective full subcategory 
of $\AQFT(\ovr{\OO(M)})$. Although this fact may be verified 
directly from the above construction of the unit, 
we do not pursue this aspect here in detail 
because this is not used in the rest of the paper. 
Nevertheless, it is worth mentioning that this observation 
justifies the interpretation of any component of the counit 
\eqref{eq:glue-data-counit} of the adjunction $\glue \dashv \data$ 
as a comparison morphism $\AAA^{\MMM} \to \AAA$ in 
$\AQFT(\ovr{\OO(M)})$ that relates a given AQFT 
$\AAA \in \AQFT(\ovr{\OO(M)})$ and the associated AQFT 
$\AAA^{\MMM} \coloneqq \glue(\data(\AAA)) \in \AQFT(\ovr{\OO(M)})$ 
constructed by gluing the local data coming from $\AAA$ 
subordinate to an open cover $\MMM$ of $M$. 
One may think of $\AAA$ as being \textit{determined by local data} 
when the comparison morphisms $\AAA^{\MMM} \to \AAA$ 
in $\AQFT(\ovr{\OO(M)})$ are isomorphisms, 
for all open covers $\MMM$ of $M$. This naturally directs attention to 
the full subcategory of $\AQFT(\ovr{\OO(M)})$ consisting of those 
AQFTs $\AAA$ that are determined by local data, namely such that 
the comparison morphisms $\AAA^{\MMM} \to \AAA$ are isomorphisms 
for all open covers $\MMM$ of $M$. The full subcategories 
of locally determined AQFTs over $M$, for all $M$, 
play a key role in \cite{BeniniGrant-StuartSchenkel_2024_HaagKastlerStacks}. 
It is explained there that such categories assemble into a $2$-functor 
(equivalently, a fibered category), which (under mild assumptions) 
is actually a stack. More explicitly, one can show that 
(by considering AQFTs defined only on relatively compact, instead of all, 
open subsets of $M$) the category of locally determined AQFTs over $M$ 
is itself determined by local data.
\end{rem}

In preparation for Section~\ref{sec:descent}, 
let us develop a more explicit description of the glued AQFT 
\begin{equation}\label{eq:glue-explicit}
\glue(\AAAA) \in \AQFT(\ovr{\OO(M)}), 
\end{equation}
for $\AAAA = \{\AAA_\alpha,\aaa_{\alpha\beta}\} \in \Desc(\MMM)$ 
any descent datum subordinate to an open cover $\MMM = \{M_\alpha\}$ 
of $M \in \Man_m$. This description relies crucially 
on Propositions~\ref{propo:astObj} and~\ref{propo:astAlg} 
and on the models for coequalizers and coproducts 
from \eqref{eq:astAlg-coeq} and \eqref{eq:astAlg-coprod}. 
\begin{enumerate}[label=(\alph*)]
\item For each $U \in \ovr{\OO(M)}$, the underlying object 
\begin{subequations}\label{eq:glue-explicit-obj}
\begin{equation}
\glue(\AAAA)(U) \in \MM
\end{equation} 
of the glued AQFT $\glue(\AAAA) \in \AQFT(\ovr{\OO(M)})$ 
is computed by the colimit of the diagram 
\begin{equation}\label{eq:glue-explicit-obj-b}
\xymatrix{
& L \Big( \coprod_\alpha R L R \big( \AAA_\alpha \vert^M \big) \Big) (U)\ar@<-2pt>[d] \ar@<2pt>[d] \\
L \Big( \coprod_{\alpha\beta} R \big( \AAA_{\alpha} \vert_{M_{\alpha\beta}} \vert^M \big) \Big) (U) \ar@<-2pt>[r] \ar@<2pt>[r] & L \Big( \coprod_\alpha R \big( \AAA_\alpha \vert^M \big) \Big) (U)
}
\end{equation}
\end{subequations}
in $\MM$. Here, the vertical parallel pair implements 
the coproduct over $\alpha$ that appears in \eqref{eq:glue-formula} 
and it comes from \eqref{eq:astAlg-coprod}; in particular, it follows from 
\eqref{eq:free-O-action} and \eqref{eq:free-forget-counit} 
that one vertical arrow is associated 
with the compositions of the AQFT colored $\ast$-operad 
$\O_{\ovr{\OO(M)}} \in \astOp$ and the other vertical arrow is 
associated with the $\O_{\ovr{\OO(M)}}$-actions of 
$\AAA_\alpha \vert^M \in \AQFT(\ovr{\OO(M)})$ for all $\alpha$. 
(Informally, one can think of the vertical parallel pair 
as being responsible for the object $\glue(\AAAA)(U) \in \MM$ to encode 
the information coming from the local AQFTs 
$\AAA_\alpha \in \AQFT(\ovr{\OO(M_\alpha)})$ for all patches $M_\alpha$.)
The horizontal parallel pair, instead, implements 
the parallel pair in \eqref{eq:glue-formula}; 
in particular, one horizontal arrow comes from the component 
$\AAA_\alpha \vert_{M_\alpha\beta} \vert^{M_\alpha} \to \AAA_\alpha$ 
in $\AQFT(\ovr{\OO(M_\alpha)})$ of the extension-restriction 
adjunction counit and the other horizontal arrow comes from 
the combination of the isomorphism 
$\aaa_{\alpha\beta}\colon \AAA_\alpha \vert_{M_{\alpha\beta}} \to \AAA_\beta \vert_{M_{\alpha\beta}}$ 
in $\AQFT(\ovr{\OO(M_{\alpha\beta})})$ with the component 
$\AAA_\beta \vert_{M_\alpha\beta} \vert^{M_\beta} \to \AAA_\beta$ 
in $\AQFT(\ovr{\OO(M_\beta)})$ of the extension-restriction 
adjunction counit. (Informally, one can think of the horizontal parallel 
pair as being responsible for the object $\glue(\AAAA)(U) \in \MM$ 
to implement the relations encoded by the comparison isomorphisms 
$\aaa_{\alpha\beta}$ for all double overlaps $M_{\alpha\beta}$.)

\item For each $U \in \ovr{\OO(M)}$ and each tuple $\und{U}$ 
in $\ovr{\OO(M)}$ of length $n \geq 0$, 
the $\O_{\ovr{\OO(M)}}$-action 
\begin{equation}\label{eq:glue-explicit-action}
\alpha_{\glue(\AAAA)}\colon \O_{\ovr{\OO(M_\alpha)}}(\substack{U\\ \und{U}}) \otimes \displaystyle\bigotimes_{i=1}^n \glue(\AAAA)(U_i) \longrightarrow \glue(\AAAA)(U)
\end{equation}
in $\MM$ of the glued AQFT $\glue(\AAAA) \in \AQFT(\ovr{\OO(M)})$ 
is obtained from the compositions of the AQFT colored 
$\ast$-operad $\O_{\ovr{\OO(M)}} \in \astOp$. 
This follows from \eqref{eq:free-O-action}. 

\item For each $U \in \ovr{\OO(M)}$, the $\ast$-involution 
\begin{equation}\label{eq:glue-explicit-inv}
\ast_{\glue(\AAAA)}\colon \glue(\AAAA)(U) \longrightarrow J(\glue(\AAAA)(U)) 
\end{equation}
in $\MM$ of the the glued AQFT 
$\glue(\AAAA) \in \AQFT(\ovr{\OO(M)})$ 
is obtained from the $\ast$-involution of the AQFT colored 
$\ast$-operad $\O_{\ovr{\OO(M)}} \in \astOp$ and 
the $\ast$-involutions of the AQFTs 
$\AAA_\alpha \in \AQFT(\ovr{\OO(M_\alpha)})$ for all $\alpha$. 
This follows from \eqref{eq:free-O-involution}. 
\end{enumerate}


\section{\label{sec:descent}Probing descent for linear AQFTs on manifolds}
From now on we work in the involutive symmetric monoidal category 
$\MM = \Vec_\bbC$ of $\bbC$-vector spaces, 
see Example~\ref{ex:Vec}. As it is customary, 
we denote the category of $\ast$-algebras, 
i.e.\ reversing $\ast$-monoids in $\Vec_\bbC$, 
see Remark~\ref{rem:AQFTsAsFunctors}, 
by $\astAlg_\bbC \coloneqq \ast\Monr(\Vec_\bbC)$. 
Our goal is to test the efficacy of the gluing construction 
for AQFTs from Section~\ref{subsec:gluing}. 

We perform this test in a simple scenario, 
namely that of $(4 \ell + 1)$-dimensional AQFTs, 
for $\ell \in \bbZ_{\geq 0}$ a non-negative number, 
i.e.\ AQFTs (in $\Vec_\bbC$) defined on the orthogonal 
category $\ovr{\Man_{4 \ell + 1}}$ of $(4 \ell + 1)$-dimensional 
oriented smooth manifolds and orientation preserving 
open embeddings, see Example~\ref{ex:Man}. 

In order to test our gluing construction, 
we consider probes $\AAA \in \AQFT(\ovr{\Man_{4 \ell + 1}})$ 
that are constructed out of geometric data. 
(For $\ell = 0$, our probes generalize the free chiral boson 
on a light ray from $2$-dimensional conformal quantum field theory.) 

First, we recall that, even in such a simple scenario, 
``naively'' gluing the local algebras 
$\AAA(M_\alpha) \in \astAlg_\bbC$ that such probe AQFTs 
$\AAA \in \AQFT(\ovr{\Man_{4 \ell + 1}})$ assign to the patches 
of an open cover $\{M_\alpha\}$ of a manifold $M$ 
typically fails to recover the global algebra 
$\AAA(M) \in \astAlg_\bbC$. 

Second, we show that gluing the local AQFTs 
$\AAA_{M_\alpha} \coloneqq \AAA \circ \iota_{M_\alpha} \in \AQFT(\ovr{\OO(M_\alpha)})$ 
according to our gluing construction from Section~\ref{subsec:gluing} 
successfully recovers the global AQFT 
$\AAA_M \coloneqq \AAA \circ \iota_M \in \AQFT(\ovr{\OO(M)})$. 
(Recall that $\iota_N\colon \ovr{\OO(N)} \to \ovr{\Man_{4 \ell + 1}}$ 
denotes the orthogonal functor that sends an open subset $U$ 
of an oriented smooth manifold $N \in \Man_{4 \ell + 1}$ to $U$ itself regarded 
as an oriented smooth manifold with the structure induced by the inclusion $U \subseteq M$, see Example~\ref{ex:slice}.) 
We interpret this fact as follows: 
the local algebras forget a lot of the information 
encoded by the original AQFT 
(including functoriality and Einstein causality), 
resulting in particular in a poorly behaved 
``naive'' gluing construction; 
in contrast to that, taking as input the local AQFTs 
$\AAA_{M_\alpha} \in \AQFT(\ovr{\OO(M_\alpha)})$, 
our gluing construction from Section~\ref{subsec:gluing} 
retains much more of the information encoded 
by the original AQFT $\AAA \in \AQFT(\ovr{\Man_{4 \ell + 1}})$, 
and this suffices to recover the global AQFT 
$\AAA_M \in \AQFT(\ovr{\OO(M)})$.

\subsection{\label{subsec:probe}Probe AQFT}
This section constructs the probes $\AAA \in \AQFT(\ovr{\Man_{4 \ell + 1}})$ 
that shall be used to test our gluing construction 
from Section~\ref{subsec:gluing}. 
The construction of the probe $\AAA$ takes 
a natural vector bundle $\sfE$ endowed with 
a natural fiber metric $\langle -,- \rangle$ 
and a natural metric connection $\nabla$ as input. 
\sk

A {\it natural vector bundle} $\sfE\colon \Man_{4 \ell + 1} \to \VBun_\bbR$ 
is a section of the evident projection functor 
$\pi\colon \VBun_\bbR \to \Man_{4 \ell + 1}$ defined on the category 
$\VBun_\bbR$ whose objects $(M,E)$ are $\bbR$-vector bundles $E$ 
of finite rank over an oriented smooth manifold $M$ and whose morphisms 
$(f,\overline{f})\colon (M_1,E_1) \to (M_2,E_2)$ 
are vector bundle maps $\overline{f}\colon E_1 \to E_2$ 
that cover an open embedding $f\colon M_1 \to M_2$ 
and that are fiber-wise isomorphisms. 
We shall use the notation $\sfE(M) = (M,E_M) \in \VBun_\bbR$ and $\sfE(f) = (f,E_f)\colon (M_1,E_{M_1}) \to (M_2,E_{M_2})$ 
in $\VBun_\bbR$. 

The prime example is the natural trivial line bundle 
$\und{\bbR}\colon \Man_1 \to \VBun_\bbR$ defined by 
$\und{\bbR}(M) \coloneqq (M,M \times \bbR) \in \VBun_\bbR$ and 
$\und{\bbR}(f) \coloneqq (f,f \times \id)\colon \und{\bbR}(M_1) \to \und{\bbR}(M_2)$ in $\VBun_\bbR$. 
More details on natural vector bundles 
can be found in \cite[App.~A]{BeniniMusanteSchenkel_2024_QuantizationLorentzian}.

For our purposes, the key feature of a natural vector bundle $\sfE$ 
is that it allows to assign to an open embedding 
$f:M_1 \to M_2$ in $\Man_1$ the pull-back morphism 
\begin{equation}
f^\ast\colon \Omega^k(M_2,E_{M_2}) \longrightarrow \Omega^k(M_1,E_{M_1})
\end{equation} 
in $\Vec_\bbR$ for smooth $k$-forms and the push-forward morphism 
\begin{equation}
f_\ast\colon \Omega^k_\cc(M_1,E_{M_1}) \longrightarrow \Omega^k_\cc(M_2,E_{M_2})
\end{equation} 
in $\Vec_\bbR$ for compactly supported smooth $k$-forms. 

In particular, one obtains 
the $\sfE$-valued $k$-forms functor 
\begin{equation}
\Omega^k(\sfE)\colon \Man_{4 \ell + 1}^\op \longrightarrow \Vec_\bbR
\end{equation}
that assigns to an oriented smooth manifold $M \in \Man_{4 \ell + 1}$ 
the vector space $\Omega^k(M,E_M) \in \Vec_\bbR$ 
of $E_M$-valued smooth $k$-forms on $M$ and 
to an open embedding $f\colon M_1 \to M_2$ in $\Man_{4 \ell + 1}$ 
the pull-back morphism $f^\ast\colon \Omega^k(M_2,E_{M_2}) \to \Omega^k(M_1,E_{M_1})$ 
in $\Vec_\bbR$. 

Similarly, one obtains the compactly supported 
$\sfE$-valued $k$-forms functor 
\begin{equation}
\Omega^k_\cc(\sfE)\colon \Man_{4 \ell + 1} \longrightarrow \Vec_\bbR
\end{equation}
that assigns to an oriented smooth manifold $M \in \Man_{4 \ell + 1}$ 
the vector space $\Omega^k_\cc(M,E_M) \in \Vec_\bbR$ 
of compactly supported $E_M$-valued smooth $k$-forms on $M$ 
and to an open embedding $f\colon M_1 \to M_2$ in $\Man_{4 \ell + 1}$ 
the push-forward morphism 
$f_\ast\colon \Omega^k(M_1,E_{M_1}) \to \Omega^k(M_2,E_{M_2})$ in $\Vec_\bbR$. 
\sk

A {\it natural fiber metric} $\langle -,- \rangle$ 
on a natural vector bundle $\sfE$ is a natural transformation 
$\langle -,- \rangle\colon \sfE \otimes \sfE \to \und{\bbR}$ 
whose components 
$\langle -,- \rangle_M\colon E_M \otimes E_M \to M \otimes \bbR$ 
are fiber metrics, i.e.\ fiber-wise non-degenerate symmetric 
bilinear forms, for all $M \in \Man_{4 \ell + 1}$. 
More details on natural fiber metrics 
can be found in \cite[App.~A]{BeniniMusanteSchenkel_2024_QuantizationLorentzian}. 

For our purposes, 
the key feature of a natural fiber metric $\langle -,- \rangle$ 
on $\sfE$ is that it allows to define, for $k = 0, \ldots, 4 \ell + 1$, 
the natural integration pairings 
\begin{equation}\label{eq:int-pairing}
(-,-)\colon \Omega^k_\cc(\sfE) \otimes \Omega^{4 \ell + 1 - k}_\cc(\sfE) \longrightarrow \bbR 
\end{equation}
as the natural transformation whose components, 
for all $M \in \Man_1$, are defined by the formula 
$(\varphi,\psi)_M \coloneqq \int_M \langle \varphi \wedge \psi \rangle_M$, 
for all $\varphi \in \Omega^k_\cc(M,E_M)$ and 
$\psi \in \Omega^{4 \ell + 1 - k}_\cc(M,E_M)$. 
(In \eqref{eq:int-pairing} $\bbR$ is interpreted as the constant 
functor $\Man_{4 \ell + 1} \ni M \mapsto \bbR \in \Vec_\bbR$.) 
\sk

A {\it natural metric connection} $\nabla$ 
on a natural vector bundle $\sfE$ endowed with 
a natural fiber metric $\langle -,- \rangle$ 
is a natural transformation 
$\nabla\colon \Omega^0(\sfE) \to \Omega^1(\sfE)$ 
whose components $\nabla_M\colon \Omega^0(M,E_M) \to \Omega^1(M,E_M)$ 
are metric connections, i.e.\ 
$\mathrm{d} \circ \langle -,- \rangle_M = \langle -,- \rangle_M \circ (\nabla_M \otimes \id + \id \otimes \nabla_M)$, 
for all $M \in \Man_{4 \ell + 1}$. 
As usual, $\nabla$ is canonically extended to $k$-forms by 
$\nabla_M (\varphi\, \alpha) \coloneqq (\nabla_M \varphi)\, \alpha + \varphi\, \dd \alpha$, 
for all $M \in \Man_{4 \ell + 1}$, $\varphi \in \Omega^0(M,E_M)$ 
and $\alpha \in \Omega^k(M)$. 
In particular, it follows that 
$\dd\, \langle \varphi \wedge \psi \rangle_M = \langle \nabla_M \varphi \wedge \psi \rangle_M + (-1)^i \langle \varphi \wedge \nabla_M \psi \rangle_M$, 
for all $\varphi \in \Omega^i(M,E_M)$ 
and $\psi \in \Omega^j(M,E_M)$. Clearly, both $\nabla$ 
and its extensions to $k$-forms preserve supports. 
In particular, they can be restricted to compactly supported sections, 
resulting in the natural transformations 
$\nabla\colon \Omega^k_\cc(\sfE) \to \Omega^{k+1}_\cc(\sfE)$. 

For our purposes, the key feature of a natural metric 
connection $\nabla$ on $(\sfE,\langle -,- \rangle)$ is that, 
for all $\varphi \in \Omega^k_\cc(M,E_M)$ 
and $\psi \in \Omega^{4 \ell - k}_\cc(M,E_M)$, one has 
\begin{equation}\label{eq:tau-asymm}
(-1)^k (\varphi, \nabla \psi) = \int_M \dd\, \langle \varphi \wedge \psi \rangle_M - (\nabla \varphi, \psi)_M = - (-1)^{(k + 1) (4 \ell - k)} (\psi, \nabla \varphi)_M = - (\psi, \nabla \varphi)_M, 
\end{equation}
where in the first step we used the compatibility between 
$\nabla_M$ and $\langle -,- \rangle_M$, in the second step 
we used Stokes' theorem and the symmetry of 
$\langle -,- \rangle_M$, and in the last step we observed that 
$4 \ell$ and $k (k + 1)$ are even numbers. 
In particular setting $k = 2 \ell$ 
in \eqref{eq:tau-asymm} entails that 
\begin{equation}\label{eq:tau}
\tau \coloneqq (-,-) \circ (\id \otimes \nabla)\colon \Omega^{2 \ell}_\cc(\sfE) \otimes \Omega^{2 \ell}_\cc(\sfE) \longrightarrow \bbR. 
\end{equation} 
is anti-symmetric and hence a Poisson structure. 
\sk

The probe AQFT 
\begin{subequations}\label{eq:probe-AQFT}
\begin{equation}
\AAA \in \AQFT(\ovr{\Man_{4 \ell + 1}})
\end{equation} 
that we shall use arises from the canonical commutation relations 
quantizing the natural Poisson structure $\tau$ from \eqref{eq:tau}. 
In the spirit of Remark~\ref{rem:AQFTsAsFunctors}, 
it is convenient to construct $\AAA$ as a $\astAlg_\bbC$-valued 
functor on $\Man_{4 \ell + 1}$ subject to the causality axiom. 

Explicitly, for all oriented smooth manifolds $M \in \Man_{4 \ell + 1}$, 
we consider the $\ast$-algebra 
\begin{equation}
\AAA(M) \coloneqq T_\bbC \Omega^{2 \ell}_\cc(M,E_M) \Big/ \big\langle \varphi \otimes \psi - \psi \otimes \varphi - i\, \tau_M(\varphi \otimes \psi) \big\rangle \in \astAlg_\bbC, 
\end{equation}
where $T_\bbC\colon \Vec_\bbC \to \astAlg_\bbC$ denotes 
the free $\ast$-algebra functor, the complexification of 
$\Omega^{2 \ell}_\cc(M,E_M) \in \Vec_\bbR$ is implicit 
and $\langle S \rangle$ denotes the two-sided 
$\ast$-ideal generated by the set $S$. 
This means that the $\ast$-algebra $\AAA(M)$ is generated 
by all $\varphi \in \Omega^{2 \ell}_\cc(M,E_M)$, subject to 
the \textit{canonical commutation relations} 
\begin{equation}
\varphi \otimes \psi - \psi \otimes \varphi = i\, \tau_M(\varphi \otimes \psi),
\end{equation}
\end{subequations}
for all $\varphi,\psi \in \Omega^{2 \ell}_\cc(M,E_M)$. 

It follows from the functoriality of $T_\bbC$ 
and the naturality of $\tau$ that the assignment 
$\Man_{4 \ell + 1} \ni M \mapsto \AAA(M) \in \astAlg_\bbC$ 
admits an evident extension to a functor 
$\AAA\colon \Man_{4 \ell + 1} \to \astAlg_\bbC$. 
Furthermore, the causality axiom holds true because 
$\tau$ vanishes for disjoint supports: 
$\tau(f_{1\ast} \varphi_1 \otimes f_{2\ast} \varphi_2) = 0$ 
for all $\varphi_1 \in \Omega^{2 \ell}_\cc(M_1,E_{M_1})$ and 
$\varphi_2 \in \Omega^{2 \ell}_\cc(M_2,E_{M_2})$ whenever 
$f_1\colon M_1 \to M \leftarrow M_2 \cocolon f_2$ is an orthogonal pair 
in $\ovr{\Man_{4 \ell + 1}}$. This concludes the construction of 
the probe AQFT $\AAA \in \AQFT(\ovr{\Man_{4 \ell + 1}})$ 
associated with a natural vector bundle $\sfE$ endowed with 
a natural fiber metric $\langle -,- \rangle$ and 
a natural metric connection $\nabla$.

\subsection{\label{subsec:alg-descent}Local algebras do not recover the global algebra}
It was observed in \cite[Appendix]{BeniniSchenkel_2019_HigherStructures} 
that for the scalar field on the circle $\bbS^1$ 
gluing the local algebras on an open cover of $\bbS^1$
recovers the global algebra on $\bbS^1$ 
if and only if the circle $\bbS^1$ itself is a member of the cover. 
This means that even very simple examples of AQFTs 
$\AAA \in \AQFT(\ovr{\Man_1})$ do not behave as cosheaves 
$\AAA\colon \Man_1 \to \astAlg_\bbC$ 
(with respect to a reasonably large class of covers). 
In more physical terms, it is impossible to recover 
the algebra of observables of a scalar field 
on the circle $\bbS^1$ from the knowledge 
of the algebras of observables subordinate to an open cover 
(unless the circle $\bbS^1$ itself is a member of the open cover). 

We shall see in Section~\ref{subsec:aqft-descent} below 
that this issue is solved by gluing the local AQFTs, 
instead of the local algebras. But before doing so, we would like 
to recall (and generalize slightly) the above-mentioned observation from 
\cite[Appendix]{BeniniSchenkel_2019_HigherStructures}. 
\sk

Consider an AQFT $\AAA \in \AQFT(\ovr{\Man_{4 \ell + 1}})$ 
on the oriented smooth manifold $M \in \Man_{4 \ell + 1}$ and an open cover 
$\MMM = \{M_\alpha\}$ of $M$. $\AAA$ assigns to each pair 
of inclusions $M_\alpha \subseteq M \supseteq M_\beta$ 
the commutative square 
\begin{equation}
\xymatrix{
\AAA(M_{\alpha\beta}) \ar[r] \ar[d] & \AAA(M_\alpha) \ar[d] \\ 
\AAA(M_\beta) \ar[r] & \AAA(M)
}
\end{equation}
in $\astAlg_\bbC$. 
As a consequence, one obtains the canonical comparison morphism 
\begin{equation}\label{eq:cosheaf-mor}
G\colon \AAA[\MMM] \coloneqq \colim \left( 
\xymatrix@C=1.5em{
\coprod_{\alpha\beta} \AAA(M_{\alpha\beta}) \ar@<2pt>[r] \ar@<-2pt>[r] & \coprod_\alpha \AAA(M_\alpha) 
}
\right) \longrightarrow \AAA(M) 
\end{equation}
in $\astAlg_\bbC$ out of the coequalizer. (Compare the latter 
displayed equation to the counit \eqref{eq:glue-data-counit} 
of the adjunction $\glue \dashv \data$, see also Remark \ref{rem:glue-data-counit}.) 
The functor $\AAA\colon \Man_{4 \ell + 1} \to \astAlg_\bbC$ 
is a cosheaf (with respect to open covers) 
when \eqref{eq:cosheaf-mor} is an isomorphism in $\astAlg_\bbC$, 
for all oriented smooth manifolds $M \in \Man_{4 \ell + 1}$ and all open covers 
$\MMM = \{M_\alpha\}$ of $M$. 
Informally, this means that gluing the local 
algebras $\AAA(M_\alpha) \in \astAlg_\bbC$ recovers 
the global algebra $\AAA(M) \in \astAlg_\bbC$. 
\sk

The colimit $\AAA[\MMM] \in \astAlg_\bbC$ 
from \eqref{eq:cosheaf-mor} admits the following 
explicit description, see Section~\ref{subsec:probe} 
and \cite[Appendix]{BeniniSchenkel_2019_HigherStructures}: 
\begin{subequations}\label{eq:glue-alg}
\begin{equation}
\AAA[\MMM] \cong T_\bbC \big( \xymatrix{\bigoplus_\alpha \Omega^{2 \ell}_\cc(M_\alpha,E_{M_\alpha})} \big) \Big/ \big\langle \iota_\alpha \varphi \otimes \iota_\alpha \psi - \iota_\alpha \psi \otimes \iota_\alpha \varphi - i\, \tau_{M_\alpha}(\varphi \otimes \psi),\; \iota_\alpha \rho - \iota_\beta \rho \big\rangle \in \astAlg_\bbC, 
\end{equation}
where $\iota_\alpha\colon \Omega^{2 \ell}_\cc(M_\alpha,E_{M_\alpha}) \to \bigoplus_\alpha \Omega^{2 \ell}_\cc(M_\alpha,E_{M_\alpha})$ 
in $\Vec_\bbR$ denotes the canonical inclusion 
of the $\alpha$-summand. 
This means that the $\ast$-algebra $\AAA[\MMM] \in \astAlg_\bbC$ is generated 
by all $\iota_\alpha \varphi \in \bigoplus_\alpha \Omega^{2 \ell}_\cc(M_\alpha, \linebreak E_{M_\alpha})$ 
subject to: 
\begin{enumerate}[label=(\roman*)]
\item the \textit{$\MMM$-subordinate canonical commutation relations} 
\begin{equation}
\iota_\alpha \varphi \otimes \iota_\alpha \psi - \iota_\alpha \psi \otimes \iota_\alpha \varphi = i \tau_{M_\alpha}(\varphi \otimes \psi)
\end{equation}
for all indices $\alpha$ and all $\varphi,\psi \in \Omega^{2 \ell}_\cc(M_\alpha,E_{M_\alpha})$,
\item the \textit{overlap relations} 
\begin{equation}
\iota_\alpha \rho = \iota_\beta \rho
\end{equation}
for all indices 
$\alpha, \beta$ and all $\rho \in \Omega^{2 \ell}_\cc(M_{\alpha\beta},E_{M_{\alpha\beta}})$.
\end{enumerate}
\end{subequations}

\begin{theo}\label{th:alg}
Let $\AAA \in \AQFT(\ovr{\Man_{4 \ell + 1}})$ be any of the probe AQFTs 
from Section~\ref{subsec:probe} and $\MMM = \{M_\alpha\}$ 
a finite open cover of an oriented smooth manifold $M$. Then the comparison 
morphism $\AAA[\MMM] \to \AAA(M)$ in $\astAlg_\bbC$ 
from \eqref{eq:cosheaf-mor} is an isomorphism if and only if 
$M$ is a member of the open cover $\MMM$, 
i.e.\ there exists an index $\gamma$ such that $M_\gamma = M$. 
\end{theo}
\begin{proof}
If $M$ is a member of the cover $\MMM$, it follows 
immediately that the comparison morphism $\AAA[\MMM] \to \AAA(M)$ 
in $\astAlg_\bbC$ is an isomorphism. 
(For instance, use the explicit description of 
$\AAA[\MMM] \in \astAlg_\bbC$ from \eqref{eq:glue-alg}.) 
\sk

To prove also the converse, mimicking 
\cite[Prop.~A.1]{BeniniSchenkel_2019_HigherStructures}, 
we first show that, 
if the morphism $\AAA[\MMM] \to \AAA(M)$ in $\astAlg_\bbC$ 
from \eqref{eq:cosheaf-mor} is an isomorphism, 
then, for all indices $\alpha, \beta$, there exists 
a ``common successor'' $\gamma$ such that $M_\alpha \subseteq M_\gamma \supseteq M_\beta$ 
contains both $M_\alpha$ and  $M_\beta$. 

To show this fact, choose a partition of unity $\{\chi_\alpha\}$ 
subordinate to the open cover $\MMM$ of $M$ and define the morphism 
\begin{equation}\label{eq:inverse}
H\colon T_\bbC\, \Omega^{2 \ell}_\cc(M,E_M) \longrightarrow \AAA[\MMM]
\end{equation}
in $\astAlg_\bbC$ as the morphism that sends each generator $\varphi \in \Omega^{2 \ell}_\cc(M,E_M)$ to $\sum_\alpha \iota_\alpha(\chi_\alpha \varphi) \in \AAA[\MMM]$.

One checks that the composition $G \circ H$ of \eqref{eq:cosheaf-mor} 
and \eqref{eq:inverse} sends a generator $\varphi$ of $T_\bbC\, \Omega^{2 \ell}_\cc(M,E_M)$
to itself, but now regarded as a generator of $\AAA(M)$. 
Taking also into account the canonical commutation 
relations from \eqref{eq:probe-AQFT}, it follows that, for $\varphi \in \Omega^{2 \ell}_\cc(M_\alpha,E_{M_\alpha})$ 
and $\psi \in \Omega^{2 \ell}_\cc(M_\beta,E_{M_\beta})$, 
the composition $G \circ H$ sends the element 
$\varphi \otimes \psi - \varphi \otimes \psi - i \tau_M(\varphi \otimes \psi)$ 
to $0 \in \AAA(M)$. 
Since $G$ is an isomorphism per hypothesis 
(in particular, it is an injective map), 
it follows that $H$ must vanish on 
$\varphi \otimes \psi - \varphi \otimes \psi - i \tau_M(\varphi \otimes \psi) \in T_\bbC \Omega^{2 \ell}_\cc(M,E_M)$, i.e.\ 
\begin{equation}\label{eq:CCR-gluealg}
H(\varphi) \otimes H(\psi) - H(\varphi) \otimes H(\psi) = i \tau_M(\varphi \otimes \psi) \in \AAA[\MMM] . 
\end{equation}

On the other hand, $\AAA[\MMM]$ implements only the $\MMM$-subordinate canonical commutation relations 
and the overlap relations from \eqref{eq:glue-alg}. 
This entails that \eqref{eq:CCR-gluealg} can hold 
for all $\varphi \in \Omega^{2 \ell}_\cc(M_\alpha,E_{M_\alpha})$ 
and $\psi \in \Omega^{2 \ell}_\cc(M_\beta,E_{M_\beta})$ 
only if the following property holds: 
there exists an index $\gamma$ such that 
$M_\alpha \subseteq M_\gamma \supseteq M_\beta$. 
Since the open cover $\MMM$ is finite, 
iterating the above-mentioned property, 
in finitely many steps we find an index $\gamma$ such that 
$M = \bigcup_\alpha M_\alpha = M_\gamma$. 
\end{proof} 

\begin{rem}
Note that an analog of Theorem \ref{th:alg} holds for 
any open cover $\MMM$ of a compact oriented smooth manifold $M \in \Man_{4 \ell + 1}$. 
Indeed, since $M$ is compact, $\MMM$ admits a finite subcover 
labeled by a finite subset $I$ of the set of indices of $\MMM$. 
Therefore, it suffices to consider $M = \bigcup_{\alpha \in I} M_\alpha$ 
in the last part of the proof of Theorem \ref{th:alg}. 
In particular, $\ell = 0$ and $M = \bbS^1$ is precisely the case 
considered in \cite[Appendix]{BeniniSchenkel_2019_HigherStructures}. 
\end{rem}

\subsection{\label{subsec:aqft-descent}Local AQFTs recover the global AQFT} 
Theorem~\ref{th:alg} entails that gluing the local 
algebras of observables often fails to recover the global one, 
even for the simple probe AQFTs constructed 
in Section~\ref{subsec:probe}. More precisely, for our probe AQFTs 
the global algebra on any oriented smooth manifold $M \in \Man_{4 \ell + 1}$ 
is recovered via gluing the local ones 
subordinate to a finite open cover $\MMM$ of $M$ if and only if 
$M$ is a member of $\MMM$. 
In contrast to that, the present section shows 
that for the probe AQFTs from Section~\ref{subsec:probe} 
the global AQFT on any oriented smooth manifold $M \in \Man_{4 \ell + 1}$ 
is recovered via gluing the local ones subordinate to any open cover 
of $M$, as explained in Section~\ref{subsec:gluing}. 
\sk

In preparation for the statement of Theorem~\ref{th:AQFT} below, 
recall from Example \ref{ex:slice} the orthogonal functor 
$\iota_M\colon \ovr{\OO(M)} \to \ovr{\Man_{4 \ell + 1}}$ 
in $\OCat$ that endows a non-empty open subset $U \subseteq M$ 
with the restriction of the orientation and smooth manifold structure of $M \in \Man_{4 \ell + 1}$. 
Combining Definitions~\ref{defi:AQFTmultifunctor} 
and~\ref{defi:AQFT} with \eqref{eq:F!F*} leads to a functor 
$\iota_M^\ast\colon \AQFT(\ovr{\Man_{4 \ell + 1}}) \to \AQFT(\ovr{\OO(M)})$ 
that restricts AQFTs on $\ovr{\Man_{4 \ell + 1}}$ to $\ovr{\OO(M)}$. 
In the statement below, we shall consider the AQFT 
$\AAA_M \coloneqq \iota_M^\ast(\AAA) \in \AQFT(\ovr{\OO(M)})$ 
that is obtained by restricting a given AQFT 
$\AAA \in \AQFT(\ovr{\Man_{4 \ell + 1}})$ along 
the orthogonal functor $\iota_M$. 

\begin{theo}\label{th:AQFT}
Let $\AAA \in \AQFT(\ovr{\Man_{4 \ell + 1}})$ be any of the probe AQFTs 
from Section~\ref{subsec:probe}, 
$M \in \Man_{4 \ell + 1}$ an oriented smooth manifold 
and $\MMM = \{M_\alpha\}$ an open cover of $M$. 
Consider the AQFT $\AAA_M \coloneqq \iota_M^\ast(\AAA) \in \AQFT(\ovr{\OO(M)})$, 
as defined in the previous paragraph. 
Then the $\AAA_M$-component $\glue(\data(\AAA_M)) \to \AAA_M$ 
in $\AQFT(\ovr{\OO(M)})$ of the counit of the adjunction 
$\glue \dashv \data\colon \AQFT(\ovr{\OO(M)}) \to \Desc(\MMM)$, 
see \eqref{eq:glue-data} and also \eqref{eq:glue-data-counit}, 
is an isomorphism. 
\end{theo}

\begin{rem}
Recalling Remark~\ref{rem:glue-data-counit}, the morphism 
$\glue(\data(\AAA_M)) \to \AAA_M$ in $\AQFT(\ovr{\OO(M)})$ 
can be interpreted as a comparison morphism that relates the given AQFT 
$\AAA_M \in \AQFT(\ovr{\OO(M)})$ to the associated AQFT 
$\glue(\data(\AAA_M)) \in \AQFT(\ovr{\OO(M)})$ consisting of 
local data subordinate to a fixed cover $\MMM$ of $M$. 
In the language of Remark~\ref{rem:glue-data-counit}, 
Theorem~\ref{th:AQFT} states that, given a probe AQFT 
$\AAA \in \AQFT(\ovr{\Man_{4 \ell + 1}})$, the AQFT 
$\AAA_M \in \AQFT(\ovr{\OO(M)})$ is determined by local data. 
In other words, for all 
$M \in \Man_{4 \ell + 1}$ and all open covers $\MMM = \{M_\alpha\}$ of $M$, 
the global AQFT $\AAA_M \in \AQFT(\ovr{\OO(M)})$ 
is recovered by gluing the local AQFTs 
$\AAA_M \vert_{M_\alpha} = \AAA_{M_\alpha} \in \AQFT(\ovr{\OO(M_\alpha)})$. 
This statement should be contrasted with Theorem~\ref{th:alg}, 
which states that gluing
the local algebras $\AAA(M_\alpha)$ fails to recover 
the global algebra $\AAA(M)$ unless 
$M \in \MMM$ is a member of the finite open cover $\MMM$ of $M$. 
\end{rem}

In preparation for the proof of Theorem~\ref{th:AQFT}, 
it is useful to provide a concrete description 
of the glued AQFT $\glue(\AAAA) \in \AQFT(\ovr{\OO(M)})$, 
where $M \in \Man_{4 \ell + 1}$ is an oriented smooth manifold, 
$\MMM = \{M_\alpha\}$ is an open cover of $M$ and 
$\AAAA = \{\AAA_\alpha,\aaa_{\alpha\beta}\} \in \Desc(\MMM)$ 
is a descent datum subordinate to $\MMM$. 
As a starting point, we consider the description at the end 
of Section~\ref{subsec:gluing}, see \eqref{eq:glue-explicit}, 
\eqref{eq:glue-explicit-obj}, \eqref{eq:glue-explicit-action} 
and \eqref{eq:glue-explicit-inv}. In particular, we shall 
take advantage of the fact that the involutive symmetric monoidal 
category $\Vec_\bbC$ is concrete. 
This makes it possible to construct AQFTs (as $\ast$-algebras 
over a colored $\ast$-operad) by listing generators and relations. 

Explicitly, the glued AQFT 
\begin{equation}
\glue(\AAAA) \in \AQFT(\ovr{\OO(M)})
\end{equation}
admits the following concrete description:
\begin{enumerate}[label=(\alph*)]
\item For each $U \in \ovr{\OO(M)}$, consider the vector space 
\begin{equation}\label{eq:glue-concrete}
\glue(\AAAA)(U) = \left( \coprod_{n \geq 0}\; \coprod_{\vert \und{U} \vert = n}\; \coprod_{\vert \und{\alpha} \vert = n} \left( \O_{\ovr{\OO(M)}} \big( \substack{U\\ \und{U}} \big) \otimes \bigotimes_{i=1}^n \AAA_{\alpha_i} \vert^M (U_i) \right)_{\Sigma_n} \middle/ \sim \right) \in \Vec_\bbC. 
\end{equation}
Before taking the quotient by the relations $\sim$, 
this is the vector space corresponding to the lower-right corner of \eqref{eq:glue-explicit-obj-b}. 
Generators (as a vector space), labeled by length 
$n \geq 0$, $n$-tuple $\und{U} = (U_1,\ldots,U_n)$ 
of open subsets of $M$ and $n$-tuple 
$\und{\alpha} = (\alpha_1,\ldots,\alpha_n)$ of indices, 
are denoted by 
\begin{equation}\label{eq:generators}
\iota_{n,\und{U},\und{\alpha}} (o \otimes \und{a}), 
\end{equation}
where $o \in \O_{\ovr{\OO(M)}} \big( \substack{U\\ \und{U}} \big)$ 
is an operation and $\und{a} = a_1 \otimes \cdots \otimes a_n$ 
consists of observables $a_i \in \AAA_{\alpha_i}\vert^M(U_i)$, 
$i=1,\ldots,n$. The subscript $(-)_{\Sigma_n}$ denotes 
the vector space of coinvariants under the $\Sigma_n$-action induced by the operadic 
permutation action and by the symmetric braiding of $\Vec_\bbC$. 
The quotient implements the relations $\sim$ listed below, 
which correspond to computing the colimit of the 
diagram \eqref{eq:glue-explicit-obj-b}. 

The vertical parallel pair is implemented by the relations 
\begin{subequations}\label{eq:relations}
\begin{equation}\label{eq:relations-action}
\iota_{k, \und{\und{U}}, \und{\und{\alpha}}} (o\, \und{o} \otimes \und{\und{a}}) =_{\sim} \iota_{n,\und{U},\und{\alpha}} (o \otimes \und{o}\, \und{\und{a}}), 
\end{equation}
where $k \coloneqq \sum_i k_i$, 
$\und{\und{U}} \coloneqq (\und{U}_1, \ldots, \und{U}_n)$, 
$\und{\und{\alpha}} \coloneqq (\und{\alpha}_1^{k_1}, \ldots, \und{\alpha}_n^{k_n})$, 
$\und{\alpha}_i^{k_i} = (\alpha_i, \ldots, \alpha_i)$ 
(constant tuple of length $k_i$), $i = 1, \ldots, n$, 
and $\und{o}\, \und{\und{a}} \coloneqq o_1\, \und{a}_1 \otimes \cdots \otimes o_n\, \und{a}_n$, 
for all lengths $n \geq 1$ and $k_i \geq 0$, $i = 1, \ldots, n$, 
all $n$-tuples $\und{U} = (U_1, \ldots, U_n)$ and $k_i$-tuples 
$\und{U}_i = (U_{i 1}, \ldots, U_{i k_i})$, $i = 1, \ldots, n$, 
of open subsets of $M$, all $n$-tuples of indices 
$\und{\alpha} = (\alpha_1, \ldots, \alpha_n)$, all operations 
$o \in \O_{\ovr{\OO(M)}} \big( \substack{U\\ \und{U}} \big)$ and 
$o_i \in \O_{\ovr{\OO(M)}} \big( \substack{U_i\\ \und{U}_i} \big)$, 
$i = 1, \ldots, n$, and all observables 
$a_{ij} \in \AAA_{\alpha_i}\vert^M(U_{ij})$, 
$i = 1, \ldots, n$, $j = 1, \ldots , k_i$. 
These relations ``balance'' the operadic composition 
$o\, \und{o} \coloneqq \gamma_{\O_{\ovr{\OO(M)}}}(o \otimes \und{o})$ 
and the $\O_{\ovr{\OO(M)}}$-actions 
$o_i\, \und{a}_i \coloneqq \alpha_{\AAA_{\alpha_i}}(o_i \otimes \und{a}_i)$, 
for all $i = 1, \ldots, n$. 

The horizontal parallel pair in \eqref{eq:glue-explicit-obj-b} 
is implemented by the relations 
\begin{equation}\label{eq:relations-glue}
\iota_{n,\und{U},\und{\alpha}} (o \otimes \und{a}) =_{\sim} \iota_{n,\und{U},\und{\beta}} (o \otimes \aaa_{\und{\alpha} \und{\beta}}\vert^M \und{a}), 
\end{equation}
\end{subequations}
for all $n \geq 1$, all $n$-tuples $\und{U} = (U_1, \ldots, U_n)$ 
of open subsets of $M$, all $n$-tuples of indices 
$\und{\alpha} = (\alpha_1, \ldots, \alpha_n)$ and $\und{\beta} = (\beta_1, \ldots, \beta_n)$, 
all operations 
$o \in \O_{\ovr{\OO(M)}} \big( \substack{U\\ \und{U}} \big)$
and all observables 
$a_i \in \AAA_{\alpha_i}\vert_{M_{\alpha_i \beta_i}}\vert^M(U_i)$, 
$i = 1, \ldots, n$.
The counits of the extension-restriction adjunctions involved in the 
above equation have been suppressed from our notation, and 
$\aaa_{\und{\alpha} \und{\beta}}\vert^M \und{a} \coloneqq \aaa_{\alpha_1 \beta_1}\vert^M a_1 \otimes \ldots \otimes \aaa_{\alpha_n \beta_n}\vert^M a_n$, 
where $\aaa_{\alpha\beta}\colon \AAA_\alpha \vert_{M_{\alpha\beta}} \to \AAA_\beta \vert_{M_{\alpha\beta}}$ are the natural isomorphisms from \eqref{eq:overlap-isos}. 
To make sense of the equation above, recall that 
our explicit model \eqref{eq:ext-res-Man} 
of the extension-restriction adjunction is such that 
$\AAA_{\alpha} \vert_{M_{\alpha \beta}} \vert^M \cong \AAA_{\alpha} \vert_{M_{\alpha \beta}} \vert^{M_\alpha} \vert^M$. 
Hence, upon acting with the appropriate extension-restriction 
adjunction counit, one obtains observables 
$a_i \in \AAA_{\alpha_i}\vert^M(U_i)$ of the AQFT 
$\AAA_{\alpha_i}\vert^M$ (left-hand side) and observables 
$\aaa_{\alpha_i \beta_i}\vert^M a_i \in \AAA_{\beta_i}\vert^M(U_i)$ 
of the AQFT $\AAA_{\beta_i}\vert^M$ (right-hand side). 

\item For each $U \in \ovr{\OO(M)}$ and each tuple $\und{U}$ 
in $\ovr{\OO(M)}$ of length $n \geq 0$, 
one has the $\O_{\ovr{\OO(M)}}$-action 
\begin{align}\label{eq:action-concrete}
\alpha_{\glue(\AAAA)}\colon \O_{\ovr{\OO(M)}}(\substack{U\\ \und{U}}) \otimes \displaystyle\bigotimes_{i=1}^n \glue(\AAAA)(U_i) &\longrightarrow \glue(\AAAA)(U), \\
o \otimes \bigotimes_{i=1}^n \iota_{k_i,\und{U}_i,\und{\alpha}_i} (o_i \otimes \und{a}_i) &\longmapsto \iota_{k,\und{\und{U}},\und{\und{\alpha}}} (o\, \und{o} \otimes \und{\und{a}}), \nn 
\end{align}
in $\Vec_\bbC$, which is constructed using the composition of 
the AQFT colored $\ast$-operad $\O_{\ovr{\OO(M)}} \in \astOp$. 
Here $k \coloneqq \sum_{i=1}^n k_i$, 
$\und{\und{U}} \coloneqq (\und{U}_1, \ldots, \und{U}_n)$, 
$\und{\und{\alpha}} \coloneqq (\und{\alpha}_1, \ldots, \und{\alpha}_n)$ and 
$o\, \und{o} \coloneqq \gamma_{\O_{\ovr{\OO(M)}}}(o \otimes \und{o})$. 
The formula above, which is given in terms of 
the generators \eqref{eq:generators}, is 
compatible with the relations \eqref{eq:relations}. 
(Recall that operadic compositions are associative.) 
Therefore, it descends to the quotient 
in \eqref{eq:glue-concrete}, thus defining the 
$\O_{\ovr{\OO(M)}}$-actions of the glued AQFT 
$\glue(\AAAA) \in \AQFT(\ovr{\OO(M)})$. 

\item For each $U \in \ovr{\OO(M)}$, one has the $\ast$-involution 
\begin{align}
\ast_{\glue(\AAAA)}\colon \glue(\AAAA)(U) &\longrightarrow J(\glue(\AAAA)(U)), \\
\iota_{n,\und{U},\und{\alpha}} (o \otimes \und{a}) &\longmapsto \iota_{n,\und{U},\und{\alpha}} \big( \ast_{\O_{\ovr{\OO(M)}}} (o) \otimes \ast_{\AAA_{\und{\alpha}}\vert^M} (\und{a}) \big), \nn 
\end{align}
in $\Vec_\bbC$, which is constructed using the $\ast$-involutions 
of the AQFT colored $\ast$-operad $\O_{\ovr{\OO(M)}} \in \astOp$, 
see Definition~\ref{defi:AQFToperad}, 
and the $\ast$-involutions of the (extended) AQFTs 
$\AAA_{\alpha_i}\vert^M \in \AQFT(\ovr{\OO(M)})$, 
for all $i = 1, \ldots, n$. Here 
$\ast_{\AAA_{\und{\alpha}}\vert^M} (\und{a}) \coloneqq \ast_{\AAA_{\alpha_1}\vert^M} (a_1) \otimes \cdots \otimes \ast_{\AAA_{\alpha_n}\vert^M} (a_n)$. 
Again, the formula above, 
which is given in terms of the generators \eqref{eq:generators}, 
is compatible with the relations \eqref{eq:relations}. 
(Recall that $\ast$-operadic compositions, operadic 
$\ast$-algebra actions and operadic $\ast$-algebra morphisms 
are compatible with $\ast$-involutions.) 
Therefore, it descends to the quotient 
in \eqref{eq:glue-concrete}, thus defining the $\ast$-involution 
of the glued AQFT $\glue(\AAAA) \in \AQFT(\ovr{\OO(M)})$. 
\end{enumerate}
 
Another useful preparatory step for the proof of 
Theorem~\ref{th:AQFT} is to present, for any AQFT 
$\AAA \in \AQFT(\ovr{\OO(M)})$, a concrete description 
of the $\AAA$-component 
\begin{subequations}\label{eq:glue-data-counit-concrete} 
\begin{equation}
\glue(\data(\AAA)) \longrightarrow \AAA
\end{equation}
in $\AQFT(\ovr{\OO(M)})$ of the counit of the adjunction 
$\glue \dashv \data\colon \AQFT(\ovr{\OO(M)}) \to \Desc(\MMM)$, 
see \eqref{eq:glue-data-counit}. This consists, for each open 
subset $U \in \ovr{\OO(M)}$, of a linear map 
\begin{align}
\glue(\data(\AAA))(U) &\longrightarrow \AAA(U), \\
\iota_{n,\und{U},\und{\alpha}} (o \otimes \und{a}) &\longmapsto o\, \und{a} \nn 
\end{align}
\end{subequations}
in $\Vec_\bbC$, where $n \geq 0$ is the length, 
$\und{U} = (U_1, \ldots, U_n)$ is an $n$-tuple of open subsets of $M$, 
$\und{\alpha} = (\alpha_1, \ldots, \alpha_n)$ is an $n$-tuple of indices, 
$o \in \O_{\ovr{\OO(M)}} \big( \substack{U\\ \und{U}} \big)$ 
is an operation and 
$\und{a} = a_1 \otimes \cdots \otimes a_n$ consists of 
observables $a_i \in \AAA \vert_{M_{\alpha_i}} \vert^M (U_i)$, 
$i = 1, \ldots, n$. Here 
$o\, \und{a} \coloneqq \alpha_{\AAA} (o \otimes \und{a})$ 
denotes the $\O_{\ovr{\OO(M)}}$-action of 
$\AAA \in \AQFT(\ovr{\OO(M)})$, the counits of the 
extension-restriction adjunctions 
$(-)\vert^M \dashv (-)\vert_{M_{\alpha_i}}$ 
being implicit in our notation. 
The formula above, which is given in terms of 
generators \eqref{eq:generators}, is compatible with the 
relations \eqref{eq:relations}. (Recall that operadic 
$\ast$-algebra actions are associative and that the coherence 
isomorphisms $\aaa_{\alpha \beta} = \id$ 
underlying the descent datum $\data(\AAA) \in \Desc(\MM)$ 
are identities by definition, see \eqref{eq:data}.) 

\begin{proof}[Proof of Theorem~\ref{th:AQFT}]
In order to prove that the $\AAA_M$-component 
of the counit of the adjunction $\glue \dashv \data$ 
(denoted below by $L\colon \glue(\data(\AAA_M)) \to \AAA_M$) 
is an isomorphism in $\AQFT(\ovr{\OO(M)})$, it suffices to construct, 
for all open subsets $U \subseteq M$, the inverse 
\begin{subequations}\label{eq:L-inverse}
\begin{equation}
L_U^{-1}\colon \AAA_M(U) \longrightarrow \glue(\data(\AAA_M))(U) 
\end{equation}
in $\astAlg_\bbC$ of the component $L_U$. (Recall from Remark 
\ref{rem:AQFTsAsFunctors} that morphisms in $\AQFT(\ovr{\OO(M)})$ 
are natural transformations between $\astAlg_\bbC$-valued functors.) 

In order to define the candidate inverse $L_U^{-1}$, recall the 
construction of the probe AQFT $\AAA \in \AQFT(\ovr{\Man_{4 \ell + 1}})$ 
from \eqref{eq:probe-AQFT} and the concrete presentation 
of the vector space underlying the glued AQFT 
$\glue(\data(\AAA_M)) \in \AQFT(\ovr{\OO(M)})$ 
from \eqref{eq:glue-concrete}. 

We define the action of the candidate inverse $L_U^{-1}$ on generators 
$\varphi \in \Omega^{2 \ell}_\cc(U,E_M)$ of the 
$\ast$-algebra $\AAA_M(U) = \AAA(U) \in \astAlg_\bbC$ by 
\begin{equation}
L_U^{-1}(\varphi) \coloneqq \sum_\alpha \iota_{1,U,\alpha} \big( \oone_U \otimes (\chi_\alpha\, \varphi) \big), 
\end{equation}
\end{subequations}
where $\{\chi_\alpha\}$ is any partition of unity 
subordinate to the given open cover $\MMM = \{M_\alpha\}$. 
(Note that we implicitly identify the vector bundle 
$E_U$ associated with $U$ (endowed with the orientation and smooth manifold structure 
induced by $M$) as a vector sub-bundle of $E_M$, 
suppressing the vector bundle map $E_U \to E_M$ 
associated with the open embedding $U \subseteq M$.)
The right-hand side makes sense because: 
(1)~the $1$-ary operation 
$\oone_U = [\id_1,\id_U] \in \O_{\ovr{\OO(M)}} \big( \substack{U\\ U} \big)$ in the AQFT colored $\ast$-operad $\O_{\ovr{\OO(M)}}$ 
is the (obvious) one given by the trivial permutation 
$\id_1 \in \Sigma_1$ and the identity morphism $\id_U\colon U \to U$ 
in $\OO(M)$, see Definition~\ref{defi:AQFToperad}; 
(2)~the section $\chi_\alpha\, \varphi \in \Omega^{2 \ell}_\cc(U \cap M_\alpha,E_{M})$ 
has compact support contained in $U \cap M_\alpha$, 
hence it represents a (linear) observable 
in $\AAA \vert_{M_\alpha} \vert^M (U)$; 
(3)~the sum is finite because $\varphi \in \Omega^{2 \ell}_\cc(U,E_{M})$ 
is compactly supported and the partition of unity 
$\{\chi_\alpha\}$ is locally finite. 
\sk 

Let us show that $L_U^{-1}(\varphi)$ does not depend 
on the choice of partition of unity $\{\chi_\alpha\}$. 
Given another partition of unity $\{\chi^\prime_\alpha\}$, one has 
\begin{align}
L_U^{-1}(\varphi) &= \sum_{\alpha,\beta} \iota_{1,U,\alpha} \big( \oone_U \otimes (\chi_\alpha\, \chi^\prime_{\beta}\, \varphi) \big) \nn \\ 
&= \sum_{\alpha,\beta} \iota_{1,U,\beta} \big( \oone_U \otimes (\chi_\alpha\, \chi^\prime_{\beta}\, \varphi) \big) \nn \\ 
&= \sum_{\beta} \iota_{1,U,\beta} \big( \oone_U \otimes (\chi^\prime_{\beta}\, \varphi) \big). 
\end{align}
In the first step we used the identity 
$\sum_\beta \chi^\prime_\beta = 1$. In the second step we used 
the relation \eqref{eq:relations-glue}. In the last step 
we used the identity $\sum_\alpha \chi_\alpha = 1$. 
\sk 

In order to check that the definition 
in \eqref{eq:L-inverse} is compatible with the quotient 
in \eqref{eq:probe-AQFT}, for all generators 
$\varphi, \psi \in \Omega^{2 \ell}_\cc(U,E_M)$ of the 
$\ast$-algebra $\AAA_M(U) = \AAA(U) \in \astAlg_\bbC$, 
we compute the commutator 
\begin{align}\label{eq:temp}
L_U^{-1}(\varphi)\, L_U^{-1}(\psi) - L_U^{-1}(\psi)\, L_U^{-1}(\varphi) 
&= \sum_{\alpha,\beta} \iota_{2,(U,U),(\alpha,\beta)} \big( (\mu_U - \mu^\op_U) \otimes \chi_\alpha\, \varphi \otimes \chi_\beta\, \psi \big) \nn \\ 
&= \sum_{\alpha = \beta} \iota_{2,(U,U),(\alpha,\beta)} \big( (\mu_U - \mu^\op_U) \otimes \chi_\alpha\, \varphi \otimes \chi_\beta\, \psi \big) \nn \\ 
&\phantom{=} + \sum_{\alpha \neq \beta} \iota_{2,(U,U),(\alpha,\beta)} \big( (\mu_U - \mu^\op_U) \otimes \chi_\alpha\, \varphi \otimes \chi_\beta\, \psi \big). 
\end{align}
Here the $2$-ary operations 
$\mu_U = [\id_2,(\id_U,\id_U)],\, \mu^\op_U = [\op_2,(\id_U,\id_U)] \in \O_{\ovr{\OO(M)}} \big( \substack{U\\ (U,U)} \big)$ 
in the AQFT colored $\ast$-operad $\O_{\ovr{\OO(M)}}$ 
are given by the identity $\id_2 \in \Sigma_2$ 
and reversing $\op_2 \in \Sigma_2$ permutations. 

We continue the computation separately for the diagonal 
$\sum_{\alpha = \beta}$ and off-diagonal 
$\sum_{\alpha \neq \beta}$ parts. 

For the diagonal part 
$\sum_{\alpha = \beta}$, one has 
\begin{align}\label{eq:diag}
\iota_{2,(U,U),(\alpha,\alpha)} \big( (\mu_U - \mu^\op_U) \otimes \chi_\alpha\, \varphi \otimes \chi_\alpha\, \psi \big) &= 
\iota_{2,(U,U),(\alpha,\alpha)} \big( \oone_U (\mu_U - \mu^\op_U) \otimes \chi_\alpha\, \varphi \otimes \chi_\alpha\, \psi \big) \nn \\ 
&= \iota_{1,U,\alpha} \big( \oone_U \otimes (\mu_U - \mu^\op_U)\, (\chi_\alpha\, \varphi \otimes \chi_\alpha\, \psi) \big) \nn \\ 
&= \iota_{1,U,\alpha} \big( \oone_U \otimes i\, \tau_{M_\alpha} (\chi_\alpha\, \varphi \otimes \chi_\alpha\, \psi) \big) \nn \\ 
&= i\, \tau_{M} (\chi_\alpha\, \varphi \otimes \chi_\alpha\, \psi). 
\end{align}
The first step uses that $o = \oone_U\, o$ for all operations 
$o \in \O_{\ovr{\OO(M)}} \big( \substack{U\\ \und{U}} \big)$; 
the second step uses the relation \eqref{eq:relations-action}; 
the third step uses the commutation rules encoded by 
$\AAA_M(U) = \AAA(U) \in \astAlg_\bbC$, see \eqref{eq:probe-AQFT}; 
the last step uses the naturality of the Poisson structure $\tau$ 
in \eqref{eq:tau}. (The unit of the $\ast$-algebra 
$\glue(\data(\AAA_M))(U) \in \astAlg_\bbC$ is suppressed 
from the notation.) 

In order to compute the off-diagonal part 
$\sum_{\alpha \neq \beta}$, we pick a partition of unity 
$\{\rho,\rho^\prime\}$ subordinate to the open cover 
$\{U \cap M_\alpha,V^\prime\}$ of 
$(U \cap M_\alpha) \cup (U \cap M_\beta)$ constructed as follows: 
take an open subset $V \subseteq U$ that includes the support 
of $\chi_\alpha\, \varphi$ and whose closure $\ovr{V}$ (in $U$) is contained 
in $U \cap M_\alpha$, i.e.\ $\supp(\chi_\alpha\, \varphi) \subseteq V$ 
and $\ovr{V} \subseteq U \cap M_\alpha$; define the open subset 
$V^\prime \coloneqq (U \cap M_\beta) \setminus \ovr{V}$; 
by construction, $\{U \cap M_\alpha,V^\prime\}$ is an open cover 
of the union $(U \cap M_\alpha) \cup (U \cap M_\beta)$ and 
$V \cap V^\prime = \emptyset$. 

With these preparations, one has 
\begin{align}\label{eq:offdiag}
\iota&_{2,(U,U),(\alpha,\beta)} \big( (\mu_U - \mu^\op_U) \otimes \chi_\alpha\, \varphi \otimes \chi_\beta\, \psi \big) \nn \\ 
&= \iota_{2,(U,U),(\alpha,\beta)} \big( (\mu_U - \mu^\op_U) \otimes \chi_\alpha\, \varphi \otimes \rho\, \chi_\beta\, \psi \big) + \iota_{2,(U,U),(\alpha,\beta)} \big( (\mu_U - \mu^\op_U) \otimes \chi_\alpha\, \varphi \otimes \rho^\prime\, \chi_\beta\, \psi \big) \nn \\ 
&= \iota_{2,(U,U),(\alpha,\alpha)} \big( (\mu_U - \mu^\op_U) \otimes \chi_\alpha\, \varphi \otimes \rho\, \chi_\beta\, \psi \big) + \iota_{2,(U,U),(\alpha,\beta)} \big( (\mu_U - \mu^\op_U) \otimes \iota_{V}^U (\chi_\alpha\, \varphi) \otimes \iota_{V^\prime}^U (\rho^\prime\, \chi_\beta\, \psi) \big) \nn \\ 
&= \iota_{1,U,\alpha} \big( \oone_U \otimes (\mu_U - \mu^\op_U)\, (\chi_\alpha\, \varphi \otimes \rho\, \chi_\beta\, \psi) \big) + \iota_{2,(U,U),(\alpha,\beta)} \big( (\mu_U - \mu^\op_U)\, (\iota_{V}^U,\iota_{V^\prime}^U) \otimes \chi_\alpha\, \varphi \otimes \rho^\prime\, \chi_\beta\, \psi \big) \nn \\ 
&= i\, \tau_{M_\alpha} (\chi_\alpha\, \varphi \otimes \rho\, \chi_\beta\, \psi) \nn \\ 
&= i\, \tau_{M} (\chi_\alpha\, \varphi \otimes \rho\, \chi_\beta\, \psi) + i\, \tau_{M} (\chi_\alpha\, \varphi \otimes \rho^\prime\, \chi_\beta\, \psi) \nn \\ 
&= i\, \tau_{M} (\chi_\alpha\, \varphi \otimes \chi_\beta\, \psi). 
\end{align}
In the first step we used the identity $\rho + \rho^\prime = 1$. 
In the second step we used the inclusion 
$\supp(\rho\, \chi_\beta\, \varphi) \subseteq U \cap M_{\alpha\beta}$ 
and the relation \eqref{eq:relations-glue} for the first summand 
and the inclusions $\supp(\chi_\alpha \varphi) \subseteq V$, 
$\supp(\rho^\prime \chi_\beta \psi) \subseteq V^\prime$ 
for the second summand. In the third step we used 
the relation \eqref{eq:relations-action} for both summands. 
In the fourth step we used the canonical commutation relations 
from \eqref{eq:probe-AQFT} for the first summand 
and the orthogonality relation $V \cap V^\prime = \emptyset$ 
to show that the second summand vanishes.  
(Indeed, recalling Definition~\ref{defi:AQFToperad}, 
for the orthogonal pair $(\iota_{V}^U,\iota_{V^\prime}^U)$ one has 
$\mu_U\, (\iota_{V}^U,\iota_{V^\prime}^U) = [\id_2,(\iota_{V}^U,\iota_{V^\prime}^U)] = [\op_2, (\iota_{V}^U,\iota_{V^\prime}^U)] = \mu^\op_U\, (\iota_{V}^U,\iota_{V^\prime}^U)$. 
Note that, here, $\iota_V^U$ denotes both the inclusion $U \subseteq V$ 
and the operation $[\id_1, \iota_V^U] \in \O(\substack{U\\ V})$.) 
In the fifth step we used the natural Poisson structure 
$\tau$ from \eqref{eq:tau} and the fact that 
$\supp(\chi_\alpha\, \varphi) \cap \supp(\rho^\prime\, \chi_\beta\, \psi) \subseteq V \cap V^\prime = \emptyset$. 
In the last step we used linearity of $\tau_M$. 

Combining \eqref{eq:diag} and \eqref{eq:offdiag} 
with \eqref{eq:temp}, one finds 
\begin{align}
L_U^{-1}(\varphi)\, L_U^{-1}(\psi) - L_U^{-1}(\psi)\, L_U^{-1}(\varphi) 
&= \sum_{\alpha = \beta} i\, \tau_{M} (\chi_\alpha\, \varphi \otimes \chi_\beta\, \psi) + \sum_{\alpha \neq \beta} i\, \tau_{M} (\chi_\alpha\, \varphi \otimes \chi_\beta\, \psi) \nn \\ 
&= \sum_{\alpha, \beta} i\, \tau_{M} (\chi_\alpha\, \varphi \otimes \chi_\beta\, \psi) \nn \\ 
&= i\, \tau_{M} (\varphi \otimes \psi). 
\end{align}
This proves the compatibility of \eqref{eq:L-inverse} 
with the quotient in \eqref{eq:probe-AQFT} and hence that 
$L_U^{-1}$ from \eqref{eq:L-inverse} is well-defined. 
\sk 

To conclude the proof, let us check that $L_U^{-1}$ 
is indeed the inverse of $L_U$. 
On generators $\varphi \in \Omega^{2 \ell}_\cc(U,E_M)$ of the 
$\ast$-algebra $\AAA_M(U) = \AAA(U) \in \astAlg_\bbC$ one has 
\begin{equation}
L_U(L_U^{-1} \big( \varphi) \big) = L_U \left( \sum_\alpha \iota_{1,U,\alpha} (\oone_U \otimes \chi_\alpha\, \varphi) \right) = \sum_\alpha \oone_U\, (\chi_\alpha\, \varphi) = \varphi. 
\end{equation}
In the first step we used the definition \eqref{eq:L-inverse} 
of $L_U^{-1}$. In the second step we used the concrete description 
of $L_U$ from \eqref{eq:glue-data-counit-concrete}. 
In the last step we used the identity $\sum_\alpha \chi_\alpha = 1$. 
This shows that $L_U\, L_U^{-1} = \id$. 
\sk 

It remains to show that $L_U^{-1}\, L_U = \id$. 
For this purpose, it is useful to make two preliminary observations. 

First, note that the morphisms $L_U^{-1}$, 
for all open subsets $U \subseteq M$, 
are the components of a natural transformation, i.e.\ 
they assemble into a morphism 
$L^{-1}\colon \AAA_M \to \glue(\data(\AAA_M))$ 
in $\AQFT(\ovr{\OO(M)})$. Indeed, given an inclusion of open subsets 
$U \subseteq U^\prime \subseteq M$, for all generators 
$\varphi \in \Omega^{2 \ell}_\cc(U,E_M)$ of the 
$\ast$-algebra $\AAA_M(U) = \AAA(U) \in \astAlg_\bbC$, one has 
\begin{align}
\iota_U^{U^\prime} \big( L_U^{-1}(\varphi) \big) &= \sum_\alpha \iota_{1,U^\prime,\alpha} \big( \iota_U^{U^\prime} \otimes (\chi_\alpha\, \varphi) \big) \nn \\ 
&= \sum_\alpha \iota_{1,U^\prime,\alpha} \Big( \oone_{U^\prime} \otimes \big( \chi_\alpha\, (\iota_{U\,\ast}^{U^\prime} \varphi) \big) \Big) \nn \\ 
&= L^{-1}_{U^\prime} \big( \AAA_M(\iota_U^{U^\prime})\, \varphi \big). 
\end{align}
In the first step we used the definition \eqref{eq:L-inverse} 
of $L_U^{-1}$ and the $\O_{\ovr{\OO(M)}}$-action on 
$\glue(\data(\AAA_M)) \in \AQFT(\ovr{\OO(M)})$, 
see \eqref{eq:action-concrete}. 
In the second step we used the relation \eqref{eq:relations-action} 
and the $\O_{\ovr{\OO(M)}}$-action on 
$\AAA_M \vert_{M_\alpha} \vert^M \in \AQFT(\ovr{\OO(M)})$, 
see Section~\ref{subsec:probe}. 
In the last step we used the definition \eqref{eq:L-inverse} 
of $L_{U^\prime}^{-1}$ and the $\O_{\ovr{\OO(M)}}$-action on 
$\AAA_M \in \AQFT(\ovr{\OO(M)})$. 
\sk 

Second, note that all observables in $\glue(\data(\AAA_M))(U)$ 
are linear combinations of elements of the form 
\begin{equation}
\iota_{n,\und{U},\und{\alpha}}(o \otimes \und{a}) = \iota_{n^\prime,\und{U}^\prime,\und{\alpha}^\prime}(o^\prime \otimes \und{\varphi}) = o^\prime \big( \iota_{1,U^\prime_1,\alpha^\prime_1}(\oone_{U^\prime_1} \otimes \varphi_1) \otimes \cdots \otimes \iota_{1,U^\prime_{n^\prime},\alpha^\prime_{n^\prime}}(\oone_{U^\prime_{n^\prime}} \otimes \varphi_{n^\prime}) \big)
\end{equation}
where $\und{\varphi} = \varphi_1 \otimes \cdots \otimes \varphi_{n^\prime}$ and each $\varphi_i \in \Omega^{2 \ell}_\cc(U^\prime_i \cap M_{\alpha^\prime_i},E_M)$ is a generator of the $\ast$-algebra 
$\AAA_M \vert_{M_{\alpha^\prime_i}} \vert^M(U^\prime_i) \in \astAlg_\bbC$. 
To check this claim, we start from \eqref{eq:generators}. 
In the first step we consider the $\ast$-algebras 
$\AAA_M \vert_{M_{\alpha_j}} \vert^M (U_j) \in \astAlg_\bbC$, 
see \eqref{eq:probe-AQFT}, and the relation \eqref{eq:relations-action}. 
In the last step we use the $\O_{\ovr{\OO(M)}}$-action on 
$\glue(\data(\AAA_M)) \in \AQFT(\ovr{\OO(M)})$ 
from \eqref{eq:action-concrete}. 
\sk 

As a consequence of the first preliminary observation, 
$L^{-1}\, L\colon \glue(\data(\AAA_M)) \to \linebreak \glue(\data(\AAA_M))$ 
is a morphism in $\AQFT(\ovr{\OO(M)})$, in particular it respects 
the $\O_{\ovr{\OO(M)}}$-action from \eqref{eq:action-concrete}. 
Hence, as a consequence of the second preliminary observation, 
it suffices to check 
the identity $L_U^{-1}\, L_U = \id$ on observables 
in $\glue(\data(\AAA_M))(U)$ of the form 
$\iota_{1,U,\beta}(\oone_U \otimes \varphi)$, 
where $\varphi \in \Omega^{2 \ell}_\cc(U \cap M_\beta,E_M)$ 
is a generator of the $\ast$-algebra 
$\AAA_M \vert_{M_{\beta}} \vert^M(U) \in \astAlg_\bbC$. 

With this in mind, one has 
\begin{align}
L_U^{-1} \Big( L_U \big( \iota_{1,U,\beta} (\oone_U \otimes \varphi) \big) \Big) &= L_U^{-1}(\oone_U\, \varphi) \nn \\ 
&= \sum_\alpha \iota_{1,U,\alpha} \big( \oone_U \otimes (\chi_\alpha\, \varphi) \big) \nn \\ 
&= \sum_\alpha \iota_{1,U,\beta} \big( \oone_U \otimes (\chi_\alpha\, \varphi) \big) \nn \\ 
&= \iota_{1,U,\beta} (\oone_U \otimes \varphi). 
\end{align}
In the first step we used the concrete description 
\eqref{eq:glue-data-counit-concrete} of $L_U$. 
In the second step we used the definition \eqref{eq:L-inverse} 
of $L_U^{-1}$. In the third step we used the relation 
\eqref{eq:relations-glue} in combination with the fact that 
the support of $\chi_\alpha\, \varphi$ is contained in 
$U \cap M_{\alpha\beta}$. 
In the last step we used the identity $\sum_\alpha \chi_\alpha = 1$. 
This shows that $L_U\, L_U^{-1} = \id$, concluding the proof. 
\end{proof}


\section*{Acknowledgments}
The work of M.B.\ is supported in part by the MIUR Excellence 
Department Project awarded to Dipartimento di Matematica, 
Università di Genova (CUP D33C23001110001) and it is fostered by 
the National Group of Mathematical Physics (GNFM-INdAM (IT)). 




\end{document}